\documentclass[12pt]{article}
\usepackage{geometry}                
\usepackage{graphicx}
\usepackage{epsfig}
\usepackage{color}
\usepackage{subfigure}
\usepackage{graphics}
\usepackage{tikz}
\usetikzlibrary{arrows}

\usepackage{amssymb}
\usepackage{amsthm}
\usepackage{amsxtra}
\usepackage{amsmath}
\usepackage{amscd}
\usepackage{epstopdf}
\usepackage{algorithmic}

\usepackage[ruled,vlined]{algorithm2e}

\usepackage{fullpage}

\usepackage[ruled]{algorithm2e}

\theoremstyle{plain} 
\newtheorem{theorem}{Theorem}[section]
\newtheorem{lemma}[theorem]{Lemma}

\newtheorem{proposition}[theorem]{Proposition}
\newtheorem{property}[theorem]{Proposition}
\newtheorem{corollary}[theorem]{Corollary}

\newtheorem{claim}[theorem]{Claim}

\theoremstyle{definition} 

\newtheorem{definition}[theorem]{Definition}

\newcommand{\E}{E(G)}
\newcommand{\V}{V(G)}
\newcommand{\Gr}{G=(\V,\E)}

\newcommand{\ChaCli}{ChainClique}

\newcommand{\intpreq}[3]{#1 \cap #3 \subseteq #2}

\newcommand{\nintpreq}[3]{#1 \cap #3 \not\subseteq #2}
\newcommand{\capc}{\mathcal{C}}
\newcommand{\siginv}{\sigma^{-1}}

\title
{Maximal cliques structure for cocomparability graphs and applications}
\author{
    J\'er\'emie Dusart\thanks{IRIF, CNRS  \&  Universit\'e Paris Diderot, Paris, France}, Michel
     Habib$^{\ast}$\thanks{Gang, Inria Paris}
    and Derek G. Corneil\thanks{Department of Computer Science,
    University of Toronto, Toronto, Ontario, Canada}
}

\date\today                                        

\begin{document}

\maketitle

\begin{abstract}
A cocomparability graph is a graph whose complement admits a transitive orientation. An interval graph is the intersection graph of a family of intervals on the real line. In this paper we investigate the relationships between interval and cocomparability graphs. This study is motivated by recent results \cite{CDH13, DusH13} that show that for some problems, the algorithm used  on interval graphs can also be used with small modifications on cocomparability graphs. Many of these algorithms are based on graph searches that preserve cocomparability orderings.  

First we propose a characterization of cocomparability graphs via a lattice structure on the set of their maximal cliques.  Using this characterization we can prove that every maximal interval subgraph of a cocomparability graph $G$  is also a maximal chordal subgraph of $G$.
Although the size of this lattice of maximal cliques can be exponential in the size of the graph, it can be used as a framework to design and prove algorithms on cocomparability graphs.
In particular  we show  that  a new graph search, namely Local Maximal Neighborhood Search (LocalMNS) leads to an $O(n+mlogn)$ time algorithm to find a 
maximal interval subgraph of a cocomparability graph. Similarly we propose a linear time algorithm to compute all simplicial vertices in a cocomparability graph. In both cases we improve on the current state of knowledge.

\end{abstract} 

\textbf{Keywords:}  (co)-comparability graphs, interval graphs, posets, maximal antichain lattices,  maximal clique lattices, graph searches.


\section{Introduction}\label{background}

This paper is devoted to the study of  \emph{cocomparability graphs}, which are the complements of comparability graphs. \emph{A comparability graph} is simply an undirected  graph that admits a transitive acyclic orientation of its edges. Comparability graphs are well-studied and arise naturally in the process of modeling real-life problems, especially those involving partial orders. For a survey see \cite{GOL,APPMM}. We also consider \emph{interval graphs} which are the intersection graphs of a family of intervals on the real line. Comparability graphs and cocomparability graphs are well-known subclasses of perfect graphs \cite{GOL}; and  interval graphs are a  well-known subclass of cocomparability graphs \cite{BLS99}.
Clearly a given cocomparability graph $G$ together with an acyclic
transitive orientation of the edges of $\overline{G}$ (the
corresponding comparability graph) can be equivalently represented by
a poset $P_G$; thus new results in any of these three areas
immediately translate to the other two areas. In this paper, we will  often omit the translations but it is important to keep in mind that they exist.

A triple $a,b,c$ of vertices forms an \emph{asteroidal triple} if the vertices are pairwise 
independent, and every pair remains connected when the third vertex and its neighborhood are 
removed from the graph. An \emph{asteroidal triple free graph} 
(AT-free for short) is a graph with no asteroidal triples. It is well-known that  AT-free graphs 
strictly contain cocomparability graphs, see \cite{GOL}.

A classical way to characterize a cocomparability graph is by  means of an \emph{umbrella-free} total ordering of its vertices. In an ordering $\sigma$ of $G$'s vertices, an \emph{umbrella} is a triple of vertices $x,~y,~z$ such that $x <_\sigma y <_\sigma z$, $xy,~yz \notin \E$, and $xz \in \E$. It has been observed in \cite{KrSt93} that a graph is a cocomparability graph if and only if it admits an umbrella-free ordering. We will also call an umbrella-free ordering a \emph{cocomp ordering}. In a similar way, interval graphs are characterized by interval orderings, where an \emph{interval ordering} $\sigma$ is an ordering of the graph's vertices that does not admit a triple of vertices $x,~y,~z$ such that $x <_\sigma y <_\sigma z$, $xy \notin \E$, and $xz \in \E$. (Notice that an interval ordering is a cocomp ordering.)  Other characterizations of interval graphs appear in theorem \ref{th:GilHof}.

The paper studies  the relationships shared by interval and cocomparability graphs and is motivated by some recents results:

\begin{itemize}
\item For the Minimum Path Cover (MPC) Problem  (a minimum set of paths such that each vertex of $G$ belongs
  to exactly one path in the set), Corneil, Dalton and  Habib showed that the greedy MPC algorithm for interval graphs, when applied to a Lexicographic Depth First Search (LDFS) cocomp ordering provides a certifying solution for cocomparability graphs (see \cite{CDH13}).  


\item For the problem of producing a cocomp ordering (assuming the graph is cocomparability) Dusart and  Habib showed that the multisweep Lexicographic Breadth First Search (LBFS)$^+$ algorithm to find an interval ordering also finds a cocomp ordering (\cite{DusH13}).  Note that $O(|V(G)|)$ LBFSs must be used in order to guarantee these results.
\end{itemize}

Other similar results
 will be surveyed in subsection \ref{list}. From these results, some natural questions arise: Do  cocomparability graphs have some kind of hidden interval structure that allows the ``lifting'' of some interval graph algorithms to cocomparability graphs? What is the role played by graph searches LBFS and LDFS and are there other searches/problems where similar results hold?

As mentioned previously, interval graphs form a strict subclass of cocomparability graphs. It is also known that every minimal triangulation of a cocomparability graph is an interval graph \cite{MohRo1996,Meister05}. In section \ref{sec:lattstruct} of this article, we will show that we can equip the set of maximal cliques of a cocomparability graph with a lattice structure where every chain of the lattice forms an interval graph.  Note that condition (iii) of  theorem \ref{th:GilHof} states that a graph $G$ is an interval graph if and only if the maximal cliques of $G$ can be linearly ordered so that for every vertex $x$, the cliques containing $x$ appear consecutively.  Thus, through the lattice, a cocomparability graph can be seen as a special composition of interval graphs.  
In particular, given a cocomparability graph $G$ with $P$ a transitive orientation of $\overline{G}$, the lattice ${\cal MA}(P)$ is formed on the set of maximal antichains of $P$ (i.e., the maximal cliques of $G$).  A graph $H=(V(H),E(H))$ with $E(H) \subseteq \E$ is a \emph{maximal chordal} (respectively \emph{interval}) \emph{subgraph} if and only if $H$ is a chordal graph and $\forall S \subseteq \E-E(H)$, $S \neq \emptyset$, $H'=(V(H),E(H)\cup S)$ is not a chordal (respectively interval) graph.  Our final result of subsection \ref{sec:maximalinter} states that 
every maximal interval subgraph of a cocomparability graph is also a maximal chordal subgraph.

In sections \ref{ALG} and \ref{sec:simplicial} we turn our attention to algorithmic applications of the theory previously developed on 
the lattice ${\cal MA}(P)$.  In section \ref{ALG} we present algorithm Chainclique which on input a graph $G$ and a total ordering
$\sigma$ of $V(G)$ returns an ordered set of cliques that collectively form an interval subgraph of $G$.  We then introduce a  new graph search (LocalMNS) that is very close to Maximal Neighborhood Search (MNS), which is a generalization of MCS, LDFS and LBFS.  We
show that Chainclique with $\sigma$ being a LocalMNS cocomparability ordering of $G$ 
returns a maximal interval subgraph of the cocomparability graph $G$;
this algorithm also gives us a way to compute a minimal interval extension of a partial order (definitions given in subsection \ref{note}). 
Section \ref{sec:simplicial} uses Chainclique to compute the set of simplicial vertices in a cocomparability graph.  

Concluding remarks appear in section \ref{CONCL}.




\subsection{Notation}\label{note}
In this article, for graphs we follow standard notation; see, for instance, \cite{GOL}. All the graphs considered here are finite, undirected,  simple and with no loops. 
An edge between vertices $u$ and $v$ is denoted by $uv$, and in this case vertices $u$ and $v$ are said to be 
\emph{adjacent}. $\overline{G}$ denotes the \emph{complement} of $\Gr$, i.e., $\overline{G} = (\V, \overline{\E})$, where $uv \in \overline{\E}$ if and only if $u\neq v$ and $uv \notin E(G)$. Let $S \subseteq \V$ be a set of vertices of $G$. 
Then, the subgraph of $G$ induced by $S$ is denoted by $G[S]$, i.e., $G[S] = (S,F)$, where for any two vertices $u, v \in S$, $uv \in F$ if and only if $uv \in \E$. 
 The set $N(v)=\{u \in \V | uv \in  \E\}$ is called the \emph{neighborhood} of the vertex $v \in \V$ in $\Gr$.  A vertex $v$ is \emph{simplicial} 
 if $G [N(v) \cup \{v\}]$ is a clique.
 An \emph{ordering} $\sigma$ of $\V$ is a permutation of $\V$ where $\sigma(i)$ is the i'th vertex in $\sigma$; $\sigma^{-1}(x)$ denotes the position of $x$ in $\sigma$.  For two vertices $u,~v$, we write that $ u <_\sigma v$ if and only if $\sigma^{-1}(u) < \sigma^{-1}(v)$. For two vertices $u$, $v \in \V$, we say that $u$ is \emph{left} (respectively \emph{right}) of $v$ in $\tau$ if $u <_\tau v$ (respectively $v <_\tau u$).
	
For partial orders we use the following notation. A \emph{partial order} (also known as a \emph{poset}) $P =(X, \leq_{P})$ is an ordered pair with a finite set $X$, the \emph{ground set} of $P$, and with a binary relation $\leq_P$ satisfying reflexivity, anti-symmetry and transitivity. For $x,~y \in X$, $x \neq y$, if $x \leq_{P} y $ or $y \leq_{P} x $ then $x,y$ are \emph{comparable}, otherwise they are \emph{incomparable} and denoted by $x \parallel_P y$.  We will also use the \emph{covering relation} in $P$ denoted by $\prec_P$, satisfying $x \prec_P y$
if and only if   $x \leq_Py$ and  $\forall z $ such that $x \leq_P z \leq_P y$  then $x=z$ or $z=y$.  In such a case we say that ``$y$ \emph{covers} $x$'' or that ``$x$ \emph {is covered by} $y$''. 

A {\it chain} (respectively an {\it antichain}) is a partial order in
which every pair (respectively no pair) is comparable.  As mentioned previously, a given cocomparability graph $G$ and a transitive
orientation of the edges of $\overline G$ can be equivalently represented by a poset $P_G$.  Note that a chain (respectively an antichain) in $P_G$ 
corresponds to an independent set (respectively a clique) in $G$.
An \emph{extension}  of a partial order $P=(X, \leq_{P})$ is a partial order $ P'= (X, \leq_{P'})$, where for $u, v \in X$, $u\leq_{P} v$ implies $u \leq_{P'} v$. In particular, if $ P'$ is a chain  then $P'$ is called a \emph{linear extension} of $P$. 
An \emph{interval extension} of a partial order is an extension that is also an interval order (\emph{interval orders} are acyclic transitive orientations of the complement of interval graphs).
$P^-$ is the poset obtained from $P$ by reversing all comparabilities.

A \emph{lattice} is a particular partial order $L=(L, \leq_{L})$\footnote{Note that we use the same symbol, namely $L$ for both the lattice and the ground set of the lattice; the exact
meaning will be clear from the context.} for which  each two-element subset $\{a,b\} \subseteq L$ has a \emph{join} (i.e., least upper bound) and a \emph{meet} (i.e., greatest lower bound), denoted by $a \vee b$ and $a \wedge b$ respectively.  This definition makes $\wedge$ and $\vee$  binary operations on $L$.
All the lattices considered here are assumed to be finite.  A \emph{distributive} lattice is one in which the operations of join and meet distribute over each other;  a \emph{modular} lattice is a lattice that satisfies the following self-dual condition: $x \leq b$ implies $x \vee (a \wedge b) = (x \vee a) \wedge b$ for all $a$.
For other definitions  on  lattices, 
the reader is referred to  \cite{Birkhoff, G68, DP02, S02, CLM12}.


\section{Lattice characterization of cocomparability graphs}\label{sec:lattstruct}

\subsection{The maximal antichain lattice of a partial order }\label{sec:mapreview}

It is known from Birkhoff \cite{Birkhoff}\footnote{Indeed Birkhoff studied the ideal lattice of a partial order, but there exists a natural bijection between ideals and antichains.} that ${\cal A}(P)$, the set of all antichains  of a partial order $P$, can be equipped with a lattice structure using the following relation between antichains: if $A,~B$ are two antichains in $P$ then $A \leq_{{\cal A}(P)} B$ if and only if $\forall a\in A$, $\exists b \in B$ with $a \leq_{P} b$. Furthermore, it is also well-known that the lattice ${\cal A}(P) = ({\cal A}(P),\leq_{{\cal A}(P)})$ is a distributive lattice. 

We now consider the relation $\leq_{{\cal MA}(P)}$, which is the restriction of the relation $\leq_{{\cal A}(P)}$ to the set of all maximal 
(with respect to set inclusion) antichains of $P$ denoted by ${\cal MA}(P)$. Let us  now review the main results known about ${\cal MA}(P)=({\cal MA}(P),\leq_{{\cal MA}(P)})$. 

The next lemma shows that the definition of $\leq_{{\cal MA}(P)}$ can be written symmetrically  in the 2 antichains $A$ and $B$, since they are maximal antichains.

\begin{lemma}\label{reversedef}\cite{Berhendt}
Let $A, B$ be two maximal antichains of a partial order $P$. 

$A \leq_{{\cal MA}(P)} B$ if and only if $\forall   a \in A$, $\exists b  \in B$ with $ a \leq_{P} b$ if and only if $\forall b \in B$, $\exists a  \in A$ with $ a \leq_{P} b$.  

\end{lemma}

Some helpful variations:

\begin{lemma}\label{variation1}
Let $A, B$ be two maximal antichains of a partial order $P$. 

$A <_{{\cal MA}(P)} B$ if and only if $\forall   a \in A$, $\exists b  \in B$ with $ a <_{P} b$ if and only if $\forall b \in B$, $\exists a  \in A$ with $ a <_{P} b$.  

\end{lemma}

\begin{lemma}\label{powerMA}
Let $A, B$ be two maximal antichains of a partial order $P$ such that $A \leq_{{\cal MA}(P)} B$.

If $x \in A$, $y \in B$ then $x \leq_{P} y$ or $x \parallel_{P} y$.
\end{lemma}

\begin{proof}
	We have two cases: either $A=B$ or $A \neq B$. In the first case, since $A$ is an antichain and $A=B$, we get $\forall x \in A$, $\forall y \in A$, if $x \neq y$ then $x \parallel_{P} y$ and if $x=y$ then $x \leq_{P} y$. 
	
	In the second case, suppose for contradiction there exists $x \in A$, $y \in B$ such that $y <_{P} x$. Since we are in the case where $A$ and $B$ are different maximal antichains, and  since $x$ and $y$ are comparable, necessarily we have $x \in A$-$B$, $y \in B$-$A$ and also $A \le_{{\cal MA}(P)} B$.  
	Applying lemma  \ref{reversedef} on $A$ and $B$, there exists $z \in A$-$B$ such that $z \le_P y$. By transitivity of $P$, we establish that $z \le_{P} x$, therefore contradicting $A$ being an antichain.
\end{proof}

Now we focus on an interesting consecutiveness property in ${\cal MA}(P)$.

\begin{proposition}\label{consecutivity}
    Let $A,$ $B,$ $C$ be three maximal antichains of a partial order P such that  $A \leq_{{\cal MA}(P)} B \leq_{{\cal MA}(P)} C$; then $\intpreq{A}{B}{C}$.
\end{proposition}

\begin{proof}
	In the case where $A=B$ we have that $\intpreq{A}{B}{C}$ and in the case $B=C$ we have that $\intpreq{A}{B}{C}$. So we can assume that $A \neq B$, $B \neq C$ and as a consequence that $A <_{{\cal MA}(P)} B <_{{\cal MA}(P)} C$. Suppose for sake of contradiction that $\nintpreq{A}{B}{C}$. So there exists $x \in (A\cap C)$-$B$. Since $x$ does not belong to $B$ there must exist some $y \in B$ comparable to $x$. Using lemma \ref{powerMA} on $A$, $B$ we establish that $x \leq_P y$. Again using lemma \ref{powerMA} on $B$, $C$ we get that $y \leq_P x$. Since $x \leq_P y$ and $y\leq_P x$,  necessarily  $y=x$. Therefore $x$ belongs to $B$ which  contradicts    $x \in (A\cap C)$-$B$.
\end{proof}

The covering relation between maximal antichains has also been characterized.

\begin{lemma}\label{defMA}\cite{Jakubik91}
	Let $A,$ $B$ be two different maximal antichains of a partial order $P$. 
	
	$A \prec_{{\cal MA}(P)} B$ if and only if  $\forall x \in A$-$B$ and $\forall y \in B$-$A$,  $x \prec_{P} y$.
		
\end{lemma}

\begin{proof}
	Suppose that $A \prec_{{\cal MA}(P)} B$ and let $y \in B$-$A$. Further suppose that $y$ does not cover some $x \in A$-$B$. 
Note that either $x \leq_P y$ or $y \parallel_P x$.  In the first case, 
$x \leq_P y$, and thus there exists $z$ such that $x \leq_P z \leq_P y$. But then consider $A'$ the set obtained from $A$ by: exchanging $x$ and $z$; deleting all vertices comparable with $z$ and adding all successors of $x$ incomparable with $z$.  $A'$ is a maximal antichain by construction and we have:
$A\leq_{{\cal MA}(P)} A' \leq_{{\cal MA}(P)} B$, a contradiction.

In the second case, $y \parallel_P x$, where $x \in A$-$B$. Let us consider $A' = \{z \in A $-$ y ~|~ z \prec_P y\} $+$y$.
We complete $A'$ as a maximal antichain $[A']$ by adding elements of $A \cup B$. Since $x, y \in [A']$, $[A'] \neq A$ and $[A'] \neq B$.
Then we have: $A<_{{\cal MA}(P)} [A'] <_{{\cal MA}(P)} B$, a contradiction to $A \prec_{{\cal MA}(P)} B$.

	Conversely, clearly we have $A\leq_{{\cal MA}(P)} B$. Let us now prove that $B$ covers $A$; if not, there exists some maximal antichain $A'$ such that $A<_{{\cal MA}(P)} A'  <_{{\cal MA}(P)} B$. Let $y \in B$-$A'$; there exists $z \in A'$ with $z <_P y$.	
Using Proposition \ref{consecutivity} we know that
$\intpreq{A}{A'}{B}$; then necessarily $z \notin A$. Thus there exists $x \in A$-$B$ with $x <_P z$, but then $y$ is not covered by $x$ a contradiction.
\end{proof}

We now introduce some new terminology in order to define the infimum and supremum on the lattice ${\cal MA}(P)$. 

\begin{definition}For a partial order $P=(X,\leq_P)$, $S \subseteq X$, $Max(S)=\{v \in S ~|~ \forall u \in S,~u \leq_P v$ or $u \parallel_P v\}$. $Max(S)$  is the set of maximal elements of the  partial order  $P(S)$ induced by $S$. In the same way, $Min(S)=\{v \in S ~|~ \forall u \in S,~v \leq_P u$ or $u \parallel_P v\}$. $Min(S)$ is the set of minimal elements of  $P(S)$. And $Inc(S)=\{x \in X-S ~|~ \forall y \in S, ~x \parallel_P y\}$ is the set of incomparable elements to $S$. 
\end{definition}

\begin{definition}
For two antichains $A,~B$ of a partial order $P=(X,\leq_P)$, let $S_{min}(A,B) = \{ x \in A$-$B ~|~ \exists y \in B$-$A$ with $x <_P y\}$ and  $S_{max}(A,B) = \{ x \in A$-$B ~|~ \exists y \in B$-$A$ with $y <_P x\}$. 
\end{definition}

Since $A, B$ are antichains, we necessarily have:  $S_{min}(A,B) \cap S_{max}(A,B) =\emptyset$.

\begin{figure}[ht] 
\centering
  \includegraphics[scale=.6]{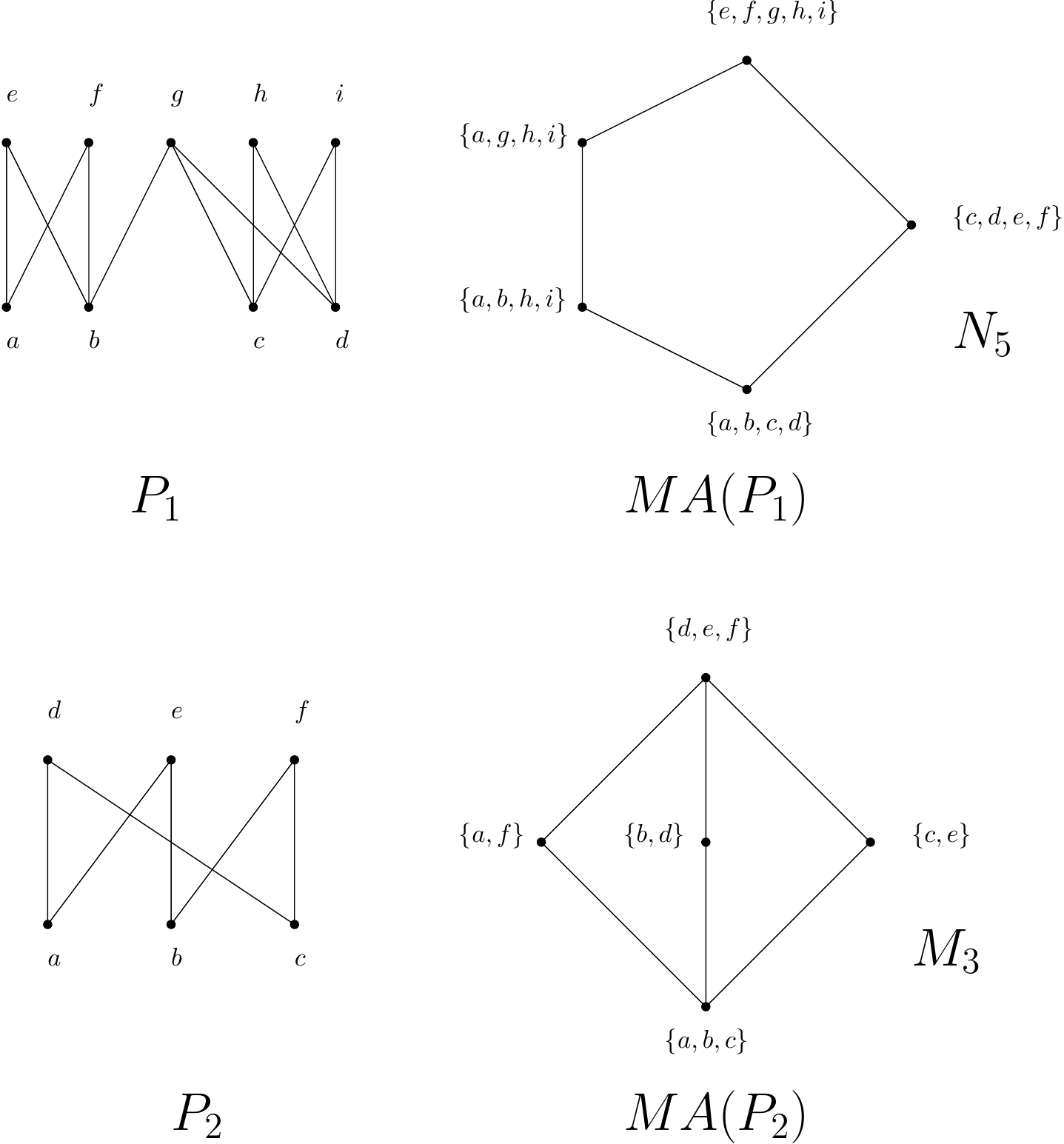}
\caption{ \small{Two orders whose maximal antichain lattices are respectively $N_{5}$ and $M_{3}$, the smallest non distributive lattices.\label{fig:distri}}}
\end{figure}  

As an example of these definitions consider the partial order $P_1$ of Figure \ref{fig:distri}. If we take the two maximal antichains
$A=\{a, g, h, i\}$ and $B=\{c, d, e, f\}$, we see that  $A \cap B =\emptyset$, $S_{min}(A,B)=\{a\}$, $S_{max}(A,B)=\{g, h, i\}$, $S_{min}(B,A)=\{c, d\}$, $S_{max}(B,A)=\{e, f\}$.  From these definitions we see:

\begin{proposition}\label{decompA}Let $A,~B$ be two maximal antichains of a partial order $P$, then $(A\cap B), S_{min}(A,B)$ and $S_{max}(A,B)$ partition $A$.
\end{proposition}

\begin{proof}
Let us consider a vertex $x$ of $A$. We have two cases: either $x \in B$ or $x \notin B$. In the first case $x \in A \cap B$. In the second case, since $B$ is a maximal antichain and $x \notin B$, there must exist $y \in B$ such that $y$ is comparable to $x$. Since $y \in B$ and $x \notin B$, we deduce that $y\neq x$. Therefore we have two cases: either $x <_P y$ or $y <_P x$. In the first case, $x \in S_{min}(A,B)$ and in the second case $x \in S_{max}(A,B)$. Thus $A=(A\cap B) \cup S_{min}(A,B) \cup S_{max}(A,B)$ and by  definition these 3 sets do not intersect.
\end{proof}

It should be noticed that in the well-known distributive lattice of antichains ${\cal A}(P)$, the infimum 
$A \wedge_{{\cal A}(P)} B =  (A\cap B) \cup S_{min}(A,B) \cup S_{min}(B,A)$ and the supremum $A \vee_{{\cal A}(P)}  B =  (A\cap B) \cup S_{max}(A,B) \cup S_{max}(B,A)$.

In this definition $A \wedge_{{\cal A}(P)}  B$ and  $A \vee_{{\cal A}(P)}  B$ are clearly antichains, but they are not maximal even if $A, B$ are maximal.  For example in $P_1$ in Figure \ref{fig:distri} we have $\{a, g, h, i\} \vee_{{\cal A}(P)} \{c, d, e, f\}=\{e, f, h, i\}$ which is not maximal.

Therefore we can now define the infimum and supremum and ${\cal MA}(P)$ as follows:

\begin{definition}
For two maximal antichains $A,~B$ of a partial order $P$, we define the binary operators $\wedge_{{\cal MA}(P)}, \vee_{{\cal MA}(P)}$:

infimum:  $A \wedge_{{\cal MA}(P)} B =  (A\cap B) \cup S_{min}(A,B) \cup S_{min}(B,A) \cup Max (Inc( (A\cap B) \cup S_{min}(A,B) \cup S_{min}(B,A) ))= A \wedge_{{\cal A}(P)} B \cup Max (Inc(A \wedge_{{\cal A}(P)} B))$.

 supremum:  $A \vee_{{\cal MA}(P)}  B =  (A\cap B) \cup S_{max}(A,B) \cup S_{max}(B,A) \cup Min (Inc( (A\cap B) \cup S_{max}(A,B) \cup S_{max}(B,A) ))= A \vee_{{\cal A}(P)} B \cup Max (Inc(A \vee_{{\cal A}(P)} B))$.
\end{definition}

Returning to the partial order $P_1$ of Figure \ref{fig:distri} where $A=\{a, g, h, i\}$ and $B=\{c, d, e, f\}$ we see that 
$Max (Inc( (A\cap B) \cup S_{min}(A,B) \cup S_{min}(B,A) ))=\{b\}$.
Therefore $\{a, g, h, i\} \wedge_{{\cal MA}(P)} \{c, d, e, f\}=\{a, b, c, d\}$. 
Similarly $\{a, g, h, i\} \vee_{{\cal MA}(P)} \{c, d, e, f\}=\{e, f, g, h, i\}$, whereas in  ${\cal A}(P)$ we have:
$\{a, g, h, i\} \vee_{{\cal A}(P)} \{c, d, e, f\}=\{e, f, h, i\} \subsetneq \{a, g, h, i\} \vee_{{\cal MA}(P)} \{c, d, e, f\}$.


Since the above supremum and  infimum  definitions are differently expressed compared to those of  \cite{Berhendt}, for completeness 
we  give a proof of  the following theorem due to Berhendt.

\begin{theorem}\label{propinf}\cite{Berhendt} Let $P$ be a partial order.
	${\cal MA}(P)=({\cal MA}(P),\wedge_{{\cal MA}(P)},\vee_{{\cal MA}(P)})$ is a lattice.
\end{theorem}

\begin{proof}
Let us first consider $A\wedge_{{\cal MA}(P)} B$.
Clearly elements of $S_{min}(A, B)$ are incomparable with elements of $S_{min}(B, A)$. Therefore $(A\cap B) \cup S_{min}(A,B) \cup S_{min}(B,A)$ is an antichain of $P$. Adding to it $Max (Inc( (A\cap B) \cup S_{min}(A,B) \cup S_{min}(B,A) ))$ completes it as a maximal antichain.

Since $A, B$ are maximal antichains, for every $x \in Inc( (A\cap B) \cup S_{min}(A,B) \cup S_{min}(B,A) )$ there exists $t \in  S_{max}(A,B)$ and $z \in  S_{max}(B,A)$ both comparable with $x$. If $t \leq_P x$, since there exists $y \in B$ such that $y \leq_P t$, it would imply 
by transitivity: $y \leq_p x$ which is impossible since $y \in S_{min}(B,A)$. Therefore $x \leq_P t$ and similarly one can obtain 
$x \leq_P z$.

Therefore we have:

$A\wedge_{{\cal MA}(P)} B \leq_{{\cal MA}(P)} A$   and $ A\wedge_{{\cal MA}(P)} B \leq_{{\cal MA}(P)} B$.

Now let us consider a maximal antichain $C$, such that:  $C \leq_{{\cal MA}(P)} A$ and $C \leq_{{\cal MA}(P)} B$.

But for every $c \in C$, there exists some $a \in A$ with  $c \leq_P a$.  If $a$ does not belong to $(A\cap B) \cup S_{min}(A,B)$ then
$a \in S_{max}(A,B)$ and so there exists $z \in Max (Inc( (A\cap B) \cup S_{min}(A,B) \cup S_{min}(B,A) ))$, with $c\leq_P a\leq_P z$.  Thus
$C \leq_{{\cal MA}(P)}A\wedge_{{\cal MA}(P)} A$. Symmetrically one can obtain $C \leq_{{\cal MA}(P)}A\wedge_{{\cal MA}(P)} B$.

Therefore this binary relation $\wedge_{{\cal MA}(P)}$ defined on maximal antichains behaves as an infimum relation on maximal antichains.

The proof is similar for $\vee_{{\cal MA}(P)}$.

\end{proof}

\begin{proposition}\label{supinf}
	Let $A$, $B$ be two maximal antichains of a partial order $P$.

	Then $(A\cup B)  \subseteq  (A \vee_{{\cal MA}(P)}  B) \cup  (A \wedge_{{\cal MA}(P)}  B)$.
\end{proposition}

\begin{proof}
		Using the definition of $A \wedge_{{\cal MA}(P)}  B$, we get that $(A \cap B) \cup S_{min}(A,B) \cup S_{min}(B,A) \subseteq (A \wedge_{{\cal MA}(P)} B)$ and symmetrically with  $A \vee_{{\cal MA}(P)}  B$ we get that $(A \cap B) \cup S_{max}(A,B) \cup S_{max}(B,A) \subseteq (A \vee_{{\cal MA}(P)} B)$.
		
		By proposition \ref{decompA}, we know that $A = (A \cap B) \cup S_{min}(A,B) \cup S_{max}(A,B)$ and $B = (A \cap B) \cup S_{min}(B,A) \cup S_{max}(B,A)$, 
thereby showing $(A\cup B)   \subseteq  A \vee_{{\cal MA}(P)}  B  \cup (A \wedge_{{\cal MA}(P)}  B)$.
\end{proof}

\begin{corollary}\label{updown}

Let $A$, $B$ be two maximal antichains of a partial order $P$ where $x \in A - B$.  Then we have two mutually exclusive cases:
either $x \in (A \wedge_{{\cal MA}(P)} B)$ or $x \in (A \vee_{{\cal MA}(P)} B)$.
\end{corollary}

\begin{proof}
Let $A$, $B$ be two maximal antichains of a partial order $P$ where $x \in A$-$B$.  Then either $x \in (A \wedge_{{\cal MA}(P)} B)$
or otherwise (using proposition \ref{supinf}), necessarily  $x \in (A \vee_{{\cal MA}(P)} B)$. 
From proposition \ref{decompA}, $A$-$B$ is partitioned into $S_{min}(A,B) \subseteq A \wedge_{{\cal MA}(P)}  B$ 
and $S_{max}(A,B) \subseteq A \vee_{{\cal MA}(P)}  B$. Therefore the two cases are mutually exclusive.

\end{proof}

\vspace{0.5cm}

There are two natural questions that arise concerning the lattice ${\cal MA}(P)$ for a given partial order $P$, namely:

\begin{itemize}
\item Does ${\cal MA}(P)$ have a particular lattice structure?
\item What is the maximum size of ${\cal MA}(P)$ given $n$, the number of elements in $P$?
\end{itemize}

The answer to the first question is ``no'' since Markowsky in \cite{Mark1} and \cite{Mark2} showed that any finite lattice is 
isomorphic to the maximal antichain lattice of a height one partial order. This result has been rediscovered by Berhendt in \cite{Berhendt}. This is summarized in the next theorem.

\begin{theorem}\cite{Berhendt,Mark1,Mark2}
	Any finite lattice is isomorphic to the lattice ${\cal MA}(P)$ of some finite partial order.
\end{theorem}

In particular, as shown in Figure \ref{fig:distri} or using the previous theorem, ${\cal MA}(P)$ is  not always distributive, 
thereby showing that ${\cal MA}(P)$ is not a sublattice of ${\cal A}(P)$ as already noticed in \cite{Berhendt}. Jakub\'ik in \cite{Jakubik91} studied for which partial orders $P$, ${\cal MA}(P)$ is modular.

For the second question, the size of ${\cal MA}(P)$ can be exponential in the number of elements of $P$. If we consider a  poset $P$ made up of $k$ disjoint chains of length 2, ${\cal MA}(P)$ has exponential size since $P$ has $2^k$ maximal antichains. The example of Figure \ref{fig:expo} shows the $k=2$ case. Furthermore Reuter showed in \cite{Reuter91} that even  the computation of the maximum length of a directed path (i.e., the height) in ${\cal MA}(P)$ is an NP-hard problem  when only $P$ is given as the input.

\begin{figure}[ht]  \centering
  \includegraphics[scale=.7]{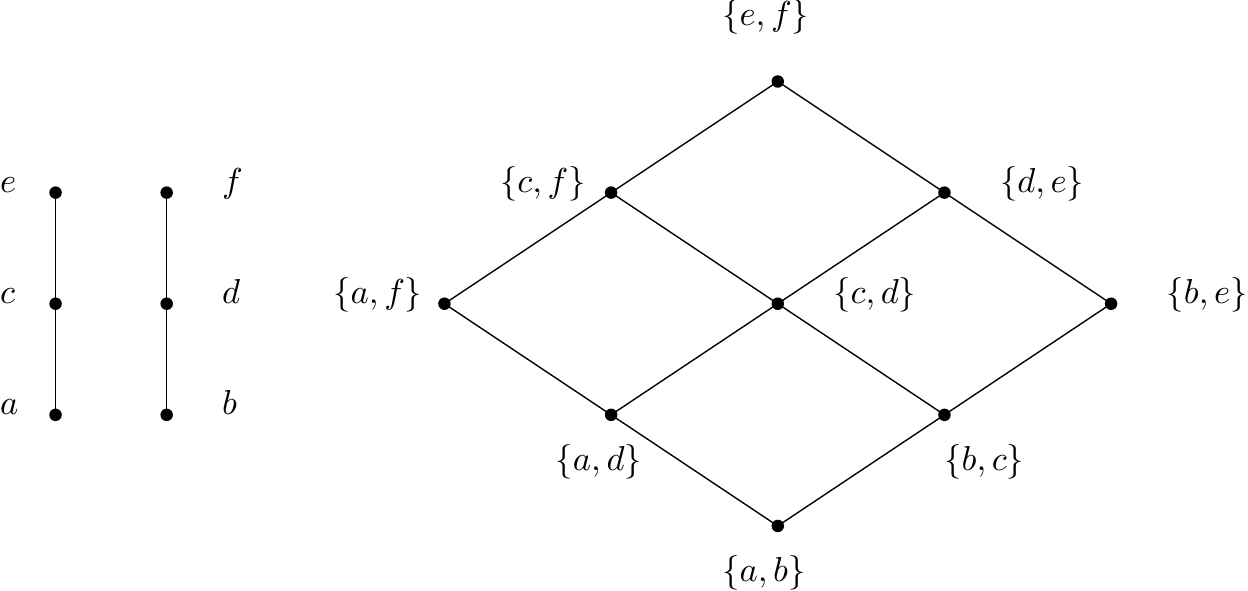}
\caption{ ${\cal MA}(P)$ for $k=2$.\label{fig:expo}}
\end{figure}

\subsection{Maximal antichain lattice and interval orders}

Following  \cite{GOL} interval graphs can be defined and characterized:

\begin{theorem}\label{th:GilHof} \cite{GH64, Lekkerkerker1962}

The following propositions are equivalent and characterize interval graphs.

(i) $G$  can be represented as  the intersection graph of a family of  intervals of the real line.

(ii) There exists a total  ordering $\tau$ of the vertices of $V$ such that $\forall x,y, z \in G$
with $x \leq_{\tau}y  \leq_{\tau} z$ and $xz \in E$ then $xy \in E$. 

(iii) The maximal cliques of $G$ can be linearly ordered such that for every vertex $x$ of $G$, the maximal cliques containing $x$ occur consecutively.


(iv) G  contains no chordless 4-cycle and is a cocomparability graph.

(v) G is chordal and has no asteroidal triple.
\end{theorem}

As mentioned in the Introduction, an ordering of the vertices satisfying  condition (ii) is called  \textit{an interval ordering} and
\emph{interval orders} are acyclic transitive orientations of the complement of interval graphs.  Therefore we
have  the following characterization theorem:

\begin{theorem} 

The following propositions are equivalent and characterize interval orders:

(i) $P$  can be represented as a left-ordering of a family of  intervals of the real line.

(ii) The successors sets are totally ordered by inclusion.

(iii) The predecessors sets are totally ordered by inclusion;


(iv) P  has a maximal antichain path. (A maximal clique path is just a maximal clique tree T, reduced to  a path).

(v) $P$ does not contain a suborder isomorphic to  \textbf{2}+ \textbf{2} (See Figure \ref{2k2}).

\end{theorem}
 \begin{figure}[ht!]
    \begin{center}
      \begin{tikzpicture}[scale=0.35]

        \coordinate(A) at (3,-5);%
        \coordinate(C) at (3,-3);%
        \coordinate(B) at (5,-5);%
        \coordinate(D) at (5,-3);%

        \draw(A)node[below]{$a$} node{$\bullet$};%
        \draw(B)node[below]{$b$} node{$\bullet$};%
        \draw(C)node[above]{$c$} node{$\bullet$};%
        \draw(D)node[above]{$d$} node{$\bullet$};%

        \draw (A)--(C);%
        \draw (B)--(D);%

              \end{tikzpicture}
    \end{center}
    \caption{A $\textbf{2}+ \textbf{2}$ partial order}\label{2k2}      
      \end{figure}
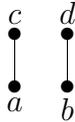

In terms  of the lattice ${\cal MA}(P)$,  condition (iv) becomes:

\begin{property}\label{chain}\cite{Berhendt}
$P$ is an interval order if and only if ${\cal MA}(P)$ is a chain.
\end{property}

This result can be complemented by:

\begin{theorem}\label{extintmin}\cite{HMPR91}
The minimal interval order extensions of $P$ are in a one-to-one correspondence with the maximal chains of ${\cal MA}(P)$.
\end{theorem}

As a consequence developed in \cite{HMPR91}, the number of minimal interval orders extensions of $P$ is a comparability invariant and is \#P-complete to compute.
For additional background information on interval orders, the reader is encouraged to consult Fishburn's monograph \cite{Fishburn85} or Trotter's survey article \cite{Trotter99}.

\subsection{Maximal cliques structure of cocomparability graphs}\label{sec:cliquecocomp}

In the following, we first characterize cocomparability graphs in terms of  a particular lattice structure on its maximal cliques and show that it extends the well-known characterization of interval graphs by a linear ordering of its maximal cliques. Then, we state some corollaries on subclasses of cocomparability graphs. 
We let ${\cal C}(G)$ denote  the set of maximal cliques of a graph $G$.


\begin{theorem}\label{cliquecocomp}
 $\Gr$ is a cocomparability graph if and only if  ${\cal C}(G)$  can be equipped with a lattice structure ${\cal L}$ satisfying:
\begin{description}
\item[(i)] For every  $A, B, C \in {\cal C}(G)$,  such that  $A \leq_{\cal L} B \leq_{\cal L } C$, then $\intpreq{A}{B}{C}$.

\item[(ii)] For every $A, B\in {\cal C}(G)$, $(A\cup B)  \subseteq  (A \vee_{\cal L}  B) \cup (A \wedge_{\cal L}  B)$.
\end{description}
\end{theorem}

\begin{proof}
	For the forward direction, let $P$ be  a partial order on $\V$ which corresponds to a transitive orientation of  $\overline{G}$.  Note that any maximal antichain of $P$ forms a maximal clique of $G$. Let ${\cal L}={\cal MA}(P)$. Proposition \ref{consecutivity} shows that the first condition is satisfied. Proposition \ref{supinf} shows that the second condition is satisfied.
		
	For the reverse direction, we will prove the following claim, which shows that if there is a lattice structure on the maximal cliques of a graph $G$ satisfying conditions (i) and (ii), then we can transitively orient $\overline{G}$ and thus $G$ is a cocomparability graph.
	
	\begin{claim}\label{giveorient}
		Let $G$ be a graph and ${\cal C}(G)$ the set of maximal cliques of $G$ such that ${\cal C}(G)$ can be equipped with a lattice structure ${\cal L}$ satisfying conditions (i) and (ii). Let $R_{\cal L}$ be a  binary relation  defined on $V(G)$ as follows:
		
		$xR_{\cal L}y$ if and only if $xy \notin \E$ and there exist 2 maximal cliques $C',$ $C''$ of G with $C' \leq_{\cal L} C''$ and $x\in C'$, $y\in C''$.
	
Then $R_{\cal L}$ is a partial order on  $V(G)$.		 
	\end{claim}

	\begin{proof}
		To show that $R_{\cal L}$ is a partial order, we start by showing that the relation is reflexive. Then we will show that  $R_{\cal L}$ is antisymmetric and finally its transitivity.

Let $x$ be a vertex of $G$. Because $G$ is simple, we have that $xx \notin \E$. Let $C_x$ be a maximal clique of $G$ that contains $x$. We have that $C_x \leq_{\cal L} C_x$ and so we deduce that $xR_{\cal L}x$ which shows the reflexivity.

If we consider different vertices $x,y \in \V$ such that $xy \notin \E$, then $x, y$ cannot be together in a maximal clique of $G$. Further, there exists at least two maximal cliques $C_x,$ $C_y$ such that $x\in C_x$ and $y\in C_y$. Assume $C_x \parallel_{{\cal L}}  C_y$.  Since $x,~y$ cannot belong together in the supremum or the infimum of  $C_x$ and $C_y$,  using  condition (ii) the supremum necessarily contains $x$ (respectively $y$) and the infimum will contain $y$ (respectively $x$). Hence, we can derive $yR_{\cal L}x$ (respectively $xR_{\cal L}y$) using the pair of maximal cliques $C_y$, $C_x \wedge_{\cal L} C_y$ (respectively the pair $C_x$, $C_x \wedge_{\cal L} C_y$). 

To show the antisymmetry of $R_{\cal L}$, let us suppose for contradiction that $xR_{\cal L}y$, $yR_{\cal L}x$ and $x \neq y$. Then there exists $C_x$, $C'_x$, $C_y$, $C'_y$  such that  $x\in C_x$, $x \in C'_x$, $y\in C_y$, $y\in C'_y$, with $C_x \leq_{{\cal L}} C_y$ and $C'_y \leq_{{\cal L}}  C'_x$. In the case where $C'_x \leq_{{\cal L}}  C_y$ then the three maximal cliques $C'_y \leq_{\cal L} C'_x \leq_{\cal L} C_y$ contradict condition (i) and if $C_y \leq_{{\cal L}}  C'_x$ then the three maximal cliques $C_x \leq_{\cal L} C_y \leq_{\cal L} C'_x$ contradict condition (i) and so we deduce that $C'_x \parallel_{{\cal L}}  C_y$. But now using condition (ii) on $C'_x$, $C_y$, we deduce that in the supremum of $C'_x$ and $C_y$ we will find either $x$ or $y$ since they cannot belong together in a maximal clique. If it is $y$ in $(C'_x \vee_{\cal L} C_y)$, we have $C'_y \leq_{\cal L} C'_x \leq_{\cal L} (C'_x \vee_{\cal L} C_y)$ that contradicts  condition (i) and similarly if it is $x$ in $(C'_x \vee_{\cal L} C_y)$, then $C_x \leq_{{\cal L}}  C_y \leq_{{\cal L}} (C'_x \vee_{\cal L} C_y) $  contradicts  condition (i). Thus if we have $xR_{\cal L}y$ and $yR_{\cal L}x$, we must have $x=y$. 

Let us now examine the transitivity of $R_{\cal L}$. Let $x,y,z$ be three different vertices such that $xR_{\cal L}y$ and $yR_{\cal L}z$.  Let us assume for contradiction that $xz \in \E$.  Therefore there exists a maximum clique $C_{xz}$ of $G$ such that  $x,z \in C_{xz}$. Let $C_y$ be a  maximal clique that contains $y$. But now because $y$ is not linked to $x$, $y$ does not belong to $C_{xz}$ and using corollary \ref{updown} on $C_{xz}$ and $C_y$ we have that either $y \in (C_{xz} \vee C_y)$ or $y \in (C_{xz} \wedge C_y)$. In the first case, by the definition of $R_{\cal L}$ we have that $zR_{\cal L}y$ and from the assumption, $yR_{\cal L}z$. So using the antisymmetry of $R_{\cal L}$ on $z$, $y$ we have that $z=y$, which contradicts our choice of $z$, $y$ being different vertices. In the second case, by the definition of $R_{\cal L}$, we have that $yR_{\cal L}x$ and from the assumption, $xR_{\cal L}y$. So using the antisymmetry on $x$, $y$, we have that $x=y$, which contradicts our choice of $x$, $y$ being different vertices. 

So assume that there exists three different vertices $x,y,z$ such that $xR_{\cal L}y$ and $yR_{\cal L}z$. Now we show that $xR_{\cal L}z$. We just have proved that $xz \notin \E$. As $xR_{\cal L}y$, there is a maximal clique $C_x$ and a maximal clique $C_y \text{ such that } x \in C_x$, $y \in C_y$ and $C_x\leq_{{\cal L}}C_y$. Let $C_z$ be a maximal clique such that $z \in C_z$. Using condition (ii) on $C_z$, $C_y$ either $z \in (C_z \wedge_{\cal L} C_y)$ or $z \in (C_z \vee_{\cal L} C_y)$. In the case where $z \in (C_z \wedge_{\cal L} C_y)$, using the definition of $R_{\cal L}$ on $z$, $y$ and the cliques $C_y$, $(C_z \wedge_{\cal L} C_y)$ we get that $zR_{\cal L}y$. Since $yR_{\cal L}z$, using the antisymmetry of $R_{\cal L}$ we get $y=z$ thereby contradicting our assumption that $y$ and $z$ are different vertices. So $z$ has to belong to $C_z \vee_{\cal L} C_y$. Now we have $C_x  \leq_{{\cal L}} C_y  \leq_{{\cal L}}  C_z \vee C_y$ and using the definition of $R_{\cal L}$ on the vertices $x$, $z$ and the cliques $C_x$, $C_z \vee_{\cal L} C_y$, we deduce that $xR_{\cal L}z$ which establishes its transitivity.
\end{proof}
In fact with Claim \ref{giveorient}, we have shown that  $R_{\cal L}$  is a transitive orientation of $\overline{G}$,
therefore $G$ is a cocomparability graph.
\end{proof}

Unfortunately as can be seen in Figure \ref{notMA}, not every  lattice ${\cal L}$ satisfying the conditions (i) and (ii)  of the previous theorem corresponds to a maximal antichain lattice ${\cal MA}(P)$ for some partial order  $P$ that gives a transitive orientation of $\overline{G}$. However by adding a simple condition, we can characterize when a lattice ${\cal L}$ is a lattice ${\cal MA}(P)$.

\noindent

 
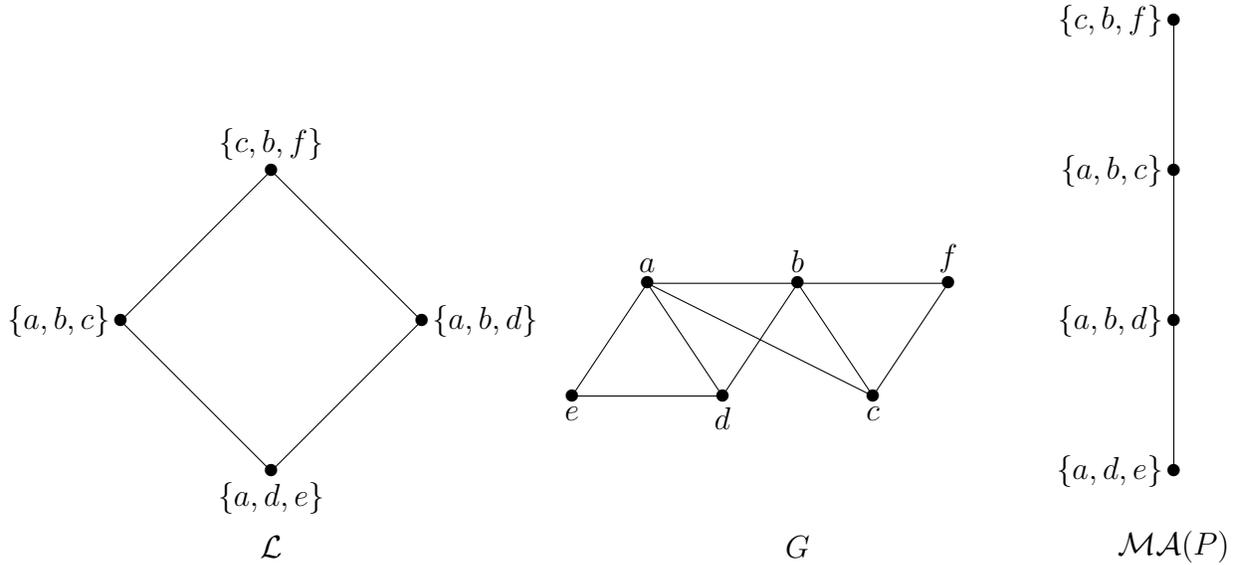
\begin{figure}[ht] 
\begin{center}
\begin{tikzpicture}
\coordinate(1) at (-4,0);
\coordinate(2) at (-6,2);
\coordinate(3) at (-2,2);
\coordinate(4) at (-4, 4);

\draw(1)node[below]{$\{a, d, e\}$} node{$\bullet$};
\draw(2)node[left]{$\{a, b, c\}$} node{$\bullet$};
\draw(3)node[right]{$\{a, b, d\}$} node{$\bullet$};
\draw(4)node[above]{$\{c, b, f\}$} node{$\bullet$};

\draw (1)--(2)--(4);
\draw (1)--(3)--(4);

\node  at (-4,-1) [scale=1] {${\cal L}$};

\coordinate(5) at (0,1);
\coordinate(6) at (2,1);
\coordinate(7) at (4,1);
\coordinate(8) at (1,2.5);
\coordinate(9) at (3,2.5);
\coordinate(10) at (5,2.5);

\draw(5)node[below]{$e$} node{$\bullet$};
\draw(6)node[below]{$d$} node{$\bullet$};
\draw(7)node[below]{$c$} node{$\bullet$};
\draw(8)node[above]{$a$} node{$\bullet$};
\draw(9)node[above]{$b$} node{$\bullet$};
\draw(10)node[above]{$f$} node{$\bullet$};

\draw (8)--(5)--(6)--(8)--(9)--(10);
\draw (8)--(7)--(10);
\draw (6)--(9)--(7);

\node  at (3, -1) [scale=1] {$G$};

\node  at (8,-1) [scale=1] {${\cal MA}(P)$};

\coordinate(11) at (8,0);
\coordinate(12) at (8,2);
\coordinate(13) at (8,4);
\coordinate(14) at (8,6);

\draw(11)node[left]{$\{a, d, e\}$} node{$\bullet$};
\draw(12)node[left]{$\{a, b, d\}$} node{$\bullet$};
\draw(13)node[left]{$\{a, b, c\}$} node{$\bullet$};
\draw(14)node[left]{$\{c, b, f\}$} node{$\bullet$};

\draw (11)--(12)--(13)--(14);

\end{tikzpicture}
\caption{\small{A graph $G$ and a lattice ${\cal L}$ on ${\cal C}(G)$ that satisfies condition (i) and (ii) of Theorem \ref{cliquecocomp}. But ${\cal L}$ is not isomorphic to the lattice ${\cal MA}(P)$ for any partial order $P$ that corresponds to a transitive orientation of $\overline G$. Since G is a prime interval graph, it has only one  transitive orientation (up to reversal) which is an interval order and its maximal antichain lattice is a chain.}}\label{notMA}
\end{center}

\end{figure}

\begin{theorem}\label{th:diffma}
	Let	$G$ be a cocomparability graph and let ${\cal L}$ be a lattice structure on ${\cal C}(G)$ satisfying conditions (i) and  (ii) of theorem \ref{cliquecocomp},
then ${\cal L}$ is isomorphic to a lattice ${\cal MA}(P)$ with $P$ a transitive orientation of $\overline{G}$ if and only if the following condition (iii) is also satisfied: 
\begin{description}
\item[(iii)]   For every $A,B \in {\cal C}(G)$, $(A \cap B) \subseteq (A \vee_{\cal L} B)$ and $(A \cap B) \subseteq (A \wedge_{\cal L} B)$.
\end{description}
\end{theorem}

\begin{proof}
	Suppose that ${\cal L}$ is isomorphic to a lattice ${\cal MA}(P)$ with $P$ a transitive orientation of $\overline{G}$. It is clear that (iii) is satisfied using the definition of $\wedge_{{\cal MA}(P)}$ and $\vee_{{\cal MA}(P)}$.

	Conversely, let us consider the partial order relation $R_{\cal L}$ defined in  claim \ref{giveorient}. We recall that $R_{\cal L}$ is defined on $\V$ as follows: $xR_{\cal L}y$ if and only if $xy \notin \E$ and there are maximal cliques $C',$ $C''$ of G with $C' \leq_{\cal L} C''$ and $x\in C'$, $y\in C''$. 
	 
	 We now prove that ${\cal L}$ is isomorphic to the lattice ${\cal MA}(R_{\cal L})$. So for this purpose we will show that for two maximal cliques $A,$ $B$, $A \leq_{\cal L} B$ if and only if $\forall x \in A$, $\exists y \in B$ with $xR_{\cal L}y$ which is the definition of ${\cal MA}(R_{\cal L})$. First we recall that since $R_{\cal L}$ is a transitive orientation of $\overline{G}$, any maximal clique of $G$ corresponds to a maximal antichain in $R_{\cal L}$. Because both ${\cal L}$ and $R_{\cal L}$ are partial orders, the relations are reflexive and the case where $A=B$ is clear.

	Let $A,B$ be two different maximal cliques of $G$ such that $A\leq_{\cal L} B$. Then $\forall x \in A$-$B$, $x$ cannot be universal to $B$-$A$ because $B$ is a maximal clique. Therefore there exists $y \in B$-$A$ such that $xy \notin \E$ and so $y$ is comparable with $x$ in $R_{\cal L}$. Furthermore we have that $x \neq y$ because $x \in A$-$B$ and $y \in B$-$A$. Using our definition of $R_{\cal L}$ on $x$, $y$ and the cliques $A$, $B$ we see that $xR_{\cal L}y$. For all $x \in A \cap B$ we also have that $xR_{\cal L}x$ and so if $A \leq_{\cal L} B$ then $\forall x \in A$, $\exists y \in B$ with $xR_{\cal L}y$.

Let $A,B$ be two different maximal cliques of $G$, such that $\forall x \in A$, $\exists y \in B$ with $xR_{\cal L}y$. For the sake of contradiction assume that $A \parallel_{\cal L} B$. Let us consider  $A \vee_{\cal L} B$. For a vertex $x \in A$-$B$, 
we know that there exists $y \in B$-$A$ such that $xR_{\cal L}y$ and so $x$ must belong to $(A \wedge_{\cal L} B)$  otherwise if $x \in (A \vee_{\cal L} B)$, using the definition of $R_{\cal L}$ on $A$, $(A \vee_{\cal L} B)$ we have that $yR_{\cal L}x$ and so $x=y$ which contradicts our choice of $x$ and $y$. But now we have that $A$-$B \subseteq (A \wedge_{\cal L} B)$ and using condition (iii) $(A \cap B) \subseteq (A \vee_{\cal L} B)$ so $A \subseteq (A \wedge_{\cal L} B)$. So either $A = (A \wedge_{\cal L} B)$ and so $A \leq_{\cal L} B$ which contradicts our choice of $A$ and $B$, or $A \subsetneq (A \wedge_{\cal L} B)$ which contradicts that A is a maximal antichain.
\end{proof}

The following corollary enlightens  the relationship between a lattice satisfying (i) and (ii) and a lattice satisfying (i),(ii) and (iii). 

\begin{corollary}
For every  lattice ${\cal L}$ associated to the maximal cliques of a cocomparability graph $G$ and  satisfying (i) and (ii) there exists a transitive orientation $R_{\cal L}$ of $\overline{G}$ such that ${\cal MA}(R_{\cal L})$ is an extension of ${\cal L}$. 
\end{corollary}
\begin{proof}
Using  claim \ref{giveorient}, for a given lattice structure ${\cal L}$ associated with a cocomparability graph and satisfying the conditions of theorem \ref{cliquecocomp}, we can define a partial order $R_{\cal L}$. From the previous proof, we know that if $A \leq_{\cal L} B$ then $\forall x \in A$, $\exists y \in B$ with $xR_{\cal L}y$ and so $A\leq_{{\cal MA}(R_{\cal L})}B$. Therefore  ${\cal MA}(R_{\cal L})$ is an extension of ${\cal L}$. 
\end{proof}

	It should be noticed that the last two theorems \ref{cliquecocomp} and \ref{th:diffma} can be easily rewritten into a characterization of comparability graphs just by exchanging maximal cliques into maximal independent sets.
Let us now study the particular case of interval graphs.

\begin{corollary}\cite{GH64}\label{intervalcharac}
    $G$ is an interval graph if and only if ${\cal C}(G)$ can be equipped with a total order $T$ satisfying for every  $C_{i}, C_{j}, C_{k}$ maximal cliques such that  $C_{i} \leq_{T} C_{j} \leq_{T} C_{k}$, then $\intpreq{C_{i}}{C_j}{C_{k}}$.
\end{corollary}

\begin{proof}
	Using the last two theorems, we know that if $G$ is a cocomparability graph then the set of maximal cliques of $G$ can be equipped with a lattice structure ${\cal L}$ satisfying conditions (i), (ii), (iii) and isomorphic to  a lattice ${\cal MA}(P)$ with $P$ a transitive orientation of $\overline{G}$.
    From property \ref{chain} ${\cal MA}(P)$ is a chain if and only if $P$ is an interval order. Since ${\cal MA}(P)$ is a chain, it should be noticed that conditions (ii) and (iii) are always satisfied and can be omitted.  Therefore only condition (i) remains. 
\end{proof}

Applied to permutation graphs\footnote{A graph is a permutation graph if and only if it is the intersection of line segments whose endpoints lie on two parallel lines.} the characterization theorems yield:

\begin{corollary}
$G$ is a permutation graph if and only if there exists a lattice structure satisfying (i), (ii)  and (iii) on the set of its maximal cliques and a  lattice structure satisfying (i), (ii) and (iii) on the set of its maximal independent sets.
\end{corollary}
\begin{proof}
We know from \cite{DM41} that $G$ is a permutation graph if and only if $G$ is a cocomparability and a comparability graph and the result follows.
\end{proof}


\subsection{Maximal chordal and interval subgraphs}\label{sec:maximalinter}

As mentioned in theorem \ref{extintmin}, for any partial order $P$ there is a bijection between maximal chains in ${\cal MA}(P)$ and minimal interval extensions of $P$. Therefore theorem \ref{th:diffma} also yields a bijection between maximal interval subgraphs of a cocomparability graph and the minimal interval extensions of $P$ (acyclic transitive orientations of $\overline{G}$). This bijection will be heavily used in the algorithms of the following sections. It should also be noticed that theorem \ref{th:diffma} gives another proof of the fact that the number of minimal interval extensions of a partial order is a comparability invariant (i.e., it does not depend on the chosen acyclic transitive orientation).

Let $G$ be a cocomparability graph and $\sigma$ a cocomp ordering of $G$.  We define $P_\sigma$ as the transitive orientation of $\overline{G}$ obtained using $\sigma$. 
For a chain $\capc=C_1 <_{{\cal MA}(P_\sigma)} C_2 \dots  <_{{\cal MA}(P_\sigma)} C_k$, $G_\capc=(\V,E(\capc))$ denotes the graph formed by the cliques $C_1,\dots,C_k$. For a vertex $x$, $N_\capc(x)$ is the neighborhood of $x$ in the graph $G_\capc$.  

\begin{proposition}\label{brief}
Consider a maximal chain of ${\cal MA}(P_\sigma)$, $\capc=C_1  \prec_{{\cal MA}(P_\sigma)} C_2 \dots \prec_{{\cal MA}(P_\sigma)} C_k$.  Such a chain forms a maximal interval subgraph of $G$.
\end{proposition}

\begin{proof}

The sequence $C_1, C_2 \dots C_k$ forms a chain of maximal cliques that respects proposition \ref{consecutivity} (consecutiveness condition). So using corollary \ref{intervalcharac}, we deduce that this chain forms a maximal interval subgraph of $G$. 
\end{proof}

Therefore, we can see a cocomparability graph as a union of interval subgraphs.

In this subsection, we now  show that for cocomparability graphs a maximal chain of ${\cal MA}(P)$ not only forms a maximal interval subgraph but also a maximal chordal subgraph. 

\begin{proposition}\label{structchain}
	Let G be a cocomparability graph and let $\sigma$ be a cocomp ordering.  Then ${\cal C}$
$ = C_1 <_{{\cal MA}(P_\sigma)} C_2 \dots <_{{\cal MA}(P_\sigma)} C_k$ is a maximal chain of ${\cal MA}(P_\sigma)$ if and only if the following conditions are satisfied \begin{itemize}
	\item $C_1$ is the set of sources of $P_\sigma$  
	\item $1<i\leq k$, $C_{i-1}  \prec_{{\cal MA}(P)} C_i$  (i.e., $C_i$ covers $C_{i-1}$)
	\item $C_k$ is the set of  sinks of $P_\sigma$
\end{itemize}
\end{proposition}

\begin{proof}

	For the forward direction, let $\capc = C_1 <_{{\cal MA}(P_\sigma)} C_2 \dots <_{{\cal MA}(P_\sigma)} C_k$ be a maximal chain of cliques of ${\cal MA}(P_\sigma)$. Let $C_S$ be the set of sources of $P_\sigma$. Since every source is incomparable with all the other sources, $C_S$ is an antichain of $P_\sigma$.  Every element that is not a source is comparable to at least one source and so $C_S$ is a maximal antichain. We now show that for every maximal antichain $A$ of $P_\sigma$, $C_S \leq_{{\cal MA}(P_\sigma)} A$. Let $A$ be a maximal antichain of $P_\sigma$; in the case $A=C_S$ then  $C_S \leq_{{\cal MA}(P_\sigma)} A$ and so we take $A\neq C_S$. For the sake of contradiction assume that $C_S \not\leq_{{\cal MA}(P_\sigma)} A$. So we have two cases: either $A <_{{\cal MA}(P_\sigma)} C_S$ or $A \parallel_{{\cal MA}(P_\sigma)} C_S$. In the first case, let $y \in C_S-A$ and using  lemma \ref{reversedef} 
	on $A$ and $C_S$ we know there exists $x \in A-C_S$ such that $x <_{P_\sigma} y$. But now because $x <_{P_\sigma} y$ we contradict the fact that $y$ is a source. In the second case, we know that there exists $(A \wedge_{{\cal MA}(P)} C_S)$ such that $(A \wedge_{{\cal MA}(P)} C_S) \leq_{{\cal MA}(P_\sigma)} C_S$. Since $A \parallel_{{\cal MA}(P_\sigma)} C_S$ we have $(A \wedge_{{\cal MA}(P)} C_S) <_{{\cal MA}(P_\sigma)} C_S$. But now we are back in the first case, which again gives us a contradiction. So for every maximal antichain $A$ of $P_\sigma$, $C_S \leq_{{\cal MA}(P_\sigma)} A$ and so we have that $C_S \leq_{{\cal MA}(P_\sigma)} C_1$. Now if $C_S \neq C_1$ then we can add $C_S$ at the beginning of the chain $C_1 <_{{\cal MA}(P_\sigma)} C_2 \dots <_{{\cal MA}(P_\sigma)} C_k$, thereby contradicting the maximality of $\capc$. Thus $C_S = C_1$.

	Now assume for contradiction that for some $1<i\leq k$, $C_i$ does not cover $C_{i-1}$. Then there exists a maximal antichain $B$ of $P_\sigma$ such that $C_{i-1} <_{{\cal MA}(P_\sigma)} B <_{{\cal MA}(P_\sigma)} C_i$. But now the chain $C_1 \leq_{{\cal MA}(P_\sigma)} \dots C_{i-1} <_{{\cal MA}(P_\sigma)} B <_{{\cal MA}(P_\sigma)} C_i \dots \leq_{{\cal MA}(P_\sigma)} C_k$ contains $\capc$ as a subchain which contradicts the maximality of $\capc$.

	Let $C_P$ be the set of sinks of $P_\sigma$. Using the same argument as in the case of the set of sources, we can deduce that for every maximal antichain $A$ of $P_\sigma$, $A \leq_{{\cal MA}(P_\sigma)} C_P$. Now if $C_P \neq C_k$ then we can add $C_P$ at the end of the chain $C_1 <_{{\cal MA}(P_\sigma)} C_2 \dots <_{{\cal MA}(P_\sigma)} C_k$, thereby contradicting the maximality of $\capc$. Thus $C_P = C_k$.

	Conversely, assume for contradiction that $\capc$ is not a maximal chain of cliques. Then we can add a maximal clique $B$ to $\capc$. There are three cases: $B$ can be added at the beginning of $\capc$; $B$ can be added at the end of $\capc$ or $B$ can be added in the middle of $\capc$. In the first case we have $B <_{{\cal MA}(P_\sigma)} C_1$, which as shown previously contradicts $C_1$ being the set of sources. Similarly the case where $C_k <_{{\cal MA}(P_\sigma)} B$ contradicts $C_k$ being the set of sinks. In the last case, we have  $C_{i-1} <_{{\cal MA}(P_\sigma)} B <_{{\cal MA}(P_\sigma)} C_i$ for some index $i$ such that  $1<i\leq k$. But this contradicts $C_{i-1} \prec_{{\cal MA}(P)} C_i$, which concludes the proof.
\end{proof}

%
%

%
%

\begin{theorem}\label{interval-chordal}

Every maximal interval subgraph of a cocomparability graph is also a maximal chordal subgraph.
\end{theorem}

\begin{proof}
Let G be a cocomparability graph and let $\sigma$ be a cocomp ordering. We just need to prove  that a maximal chain of ${\cal MA}(P_\sigma)$ is a maximal interval subgraph and a maximal chordal subgraph.

By proposition \ref{brief}, given a maximal chain of ${\cal MA}(P_\sigma)$: 
$\capc = C_1 \prec_{{\cal MA}(P_\sigma)} C_2 \dots \prec_{{\cal MA}(P_\sigma)} C_k$; then $G_\capc$ is a maximal interval subgraph.

	Assume for contradiction that $G_\capc$ is not a maximal chordal subgraph. Let $S$ be a set of edges such that the graph $H=(\V,E(\capc)\cup S)$ is a maximal chordal subgraph. In the proof, we will carefully choose an edge $uv \in S$ and show that we can find an induced path from $u$ to $v$ of length at least 3 in $H$. Therefore it will prove that $G_\capc$ is a maximal chordal subgraph. Since interval graphs are a subclass of chordal graphs and $G_\capc$ is an interval graph, we will deduce that $G_\capc$ is also a maximal interval subgraph.

	We start by proving two claims.


	\begin{claim}\label{middleclique}
		Let $G$ be a cocomparability graph, let $\sigma$ be a cocomp ordering and let $u,v$ be two vertices of $G$ such that $uv \in E$. 	
	
		If $C_u$, $C_v$ are maximal cliques of $G$ such that $u \in C_u$, $v \notin C_u$, $v \in C_v$, $u \notin C_v$,  $C_{u} <_{{\cal MA}(P_\sigma)} C_v$ then there exists a maximal clique $C_{uv}$ such that $u,v \in C_{uv}$ and $C_{u} <_{{\cal MA}(P_\sigma)} C_{uv} <_{{\cal MA}(P_\sigma)} C_v$.
	\end{claim}

	\begin{proof}

        Since $uv \in E$ there must exist at least one  maximal clique $C_{uv}^0$ of $G$ that contains $u$ and $v$. We now define $C_{uv}^1=C_{uv}^0 \vee_{{\cal MA}(P_\sigma)} C_u$ and show that $u$, $v$ belong to $C_{uv}^1$. We have three cases: $C_{u} <_{{\cal MA}(P_\sigma)} C_{uv}^0$; $C_{uv}^0 <_{{\cal MA}(P_\sigma)} C_{u}$; $C_{u} \parallel_{{\cal MA}(P_\sigma)} C_{uv}^0$. In the first case, we see that $C_{uv}^1 = C_{uv}^0$ and so $u$, $v$ belong to $C_{uv}^1$. In the second case, $v$ belongs to $C_{uv}^0$ and $C_v$ but not to $C_u$ and so $C_{uv}^0 <_{{\cal MA}(P_\sigma)} C_{u} <_{{\cal MA}(P_\sigma)} C_{v}$ contradicts proposition \ref{consecutivity} (consecutiveness condition). Therefore this case cannot happen. In the last case, using the definition of $\vee_{{\cal MA}(P_\sigma)}$ on $C_{uv}^0$ and $C_u$, we see that $u$ must belong to $C_{uv}^1$ because it belongs to $C_{uv}^0 \cap C_u$. Using again the definition of $\vee_{{\cal MA}(P_\sigma)}$  on $C_{uv}^0$ and $C_u$, we see that $v$ must belong to $C_{uv}^1$ otherwise $v$ would have to belong to $C_{uv}^0 \wedge_{{\cal MA}(P_\sigma)} C_u$ and  $(C_{uv}^0 \wedge_{{\cal MA}(P_\sigma)} C_u) <_{{\cal MA}(P_\sigma)} C_u <_{{\cal MA}(P_\sigma)} C_v$ would contradict proposition \ref{consecutivity} (consecutiveness condition).   Thus $u, v$ belong to $C_{uv}^1$.
	
	We finish the proof of the claim by showing that  $C_{uv}=C_{uv}^1 \wedge_{{\cal MA}(P_\sigma)} C_v$ satisfies $u$, $v \in C_{uv}$ and  $C_{u} <_{MA(P_\sigma)} C_{uv} <_{{\cal MA}(P_\sigma)} C_v$. From the choice of $C_{uv}$ we know $C_{uv} <_{{\cal MA}(P_\sigma)} C_v$. We have three cases: $C_{uv}^1 <_{{\cal MA}(P_\sigma)} C_{v}$; $C_{v} <_{{\cal MA}(P_\sigma)} C_{uv}^1$; $C_{uv}^1 \parallel_{{\cal MA}(P_\sigma)} C_{v}$. In the first case, we see that $C_{uv} = C_{uv}^1$ and so $u, v \in C_{uv}$ and  $C_{u} <_{{\cal MA}(P_\sigma)} C_{uv}$.  In the second case, $u$ belongs to $C_{uv}^1$ and $C_u$ but not to $C_v$ and so $C_{u} <_{{\cal MA}(P_\sigma)} C_{v} <_{{\cal MA}(P_\sigma)} C_{uv}^1$ contradicts proposition \ref{consecutivity} (consecutiveness condition). Therefore this case cannot happen. In the last case, using the definition of $\wedge_{{\cal MA}(P_\sigma)}$  on $C_{uv}^1$ and $C_v$ we see that $v$ must belong to $C_{uv}$ since it belongs to $C_{uv}^1 \cap C_v$. Using theorem \ref{cliquecocomp} on $C_{uv}^1$ and $C_u$, we get that $u$ must belong to $C_{uv}$ otherwise $u$ would have to belong to $(C_{uv}^1 \vee_{{\cal MA}(P_\sigma)} C_v)$ and  $C_{u} <_{{\cal MA}(P_\sigma)} C_v <_{{\cal MA}(P_\sigma)} (C_{uv}^1 \vee_{{\cal MA}(P_\sigma)} C_v)$ would contradict proposition \ref{consecutivity} (consecutiveness condition). Since $C_{uv}$ is defined as $C_{uv}^1 \wedge_{{\cal MA}(P_\sigma)} C_v$ by the definition of the lattice $C_{u} <_{{\cal MA}(P_\sigma)} C_{uv}$.
\end{proof}

We now introduce some terminology.
Let ${\cal C}=C_1 \prec_{{\cal MA}(P_\sigma)} C_2 \dots \prec_{{\cal MA}(P_\sigma)} C_k$ be a maximal chain of ${\cal MA}(P_\sigma)$; for every vertex $x \in V(G)$ we define  $first_{\cal C}[x]$ (respectively $last_{\cal C}[x]$) as the first (respectively last) index of a clique of $\capc$ that contains $x$.

Since a maximal interval subgraph is obviously a spanning subgraph,  these functions are well-defined. Furthermore when there is no ambiguity on $\cal C$, we simply denote these values by  $first[x]$ and $last[x]$.

	\begin{claim}
Let $G$ be a cocomparability graph, $\sigma$ be a cocomp ordering, $\capc= C_1 \prec_{{\cal MA}(P_\sigma)} C_2 \dots \prec_{{\cal MA}(P_\sigma)} C_k$ be a maximal chain of ${\cal MA}(P_\sigma)$ and $u,~v$ be two vertices of $G$ such that $uv \in E$.

       If $last[u] < first[v]$ then $\exists x,~y$ such that $first[x]\leq last[u] < first[y] \leq  last[x] < first[v]$, $last[x] \leq last[y]$ and $yv \in E$.

	\end{claim}

	\begin{proof}
	Let $C_u=C_{last[u]}$ and $C_v=C_{first[v]}$. Since $u \in N(v)-N_\capc(v)$, we deduce that $uv \notin E(\capc)$. Furthermore, since $u <_\sigma v$, $C_u,~C_v \in\capc$, we see that $C_{u} <_{{\cal MA}(P_\sigma)} C_v$. Using the previous claim on $C_u$, $C_v$, we deduce there exists $C_{uv}$ a maximal clique of $G$ such that  $C_{u} <_{{\cal MA}(P_\sigma)} C_{uv} <_{{\cal MA}(P_\sigma)} C_v$. Since $C_1,\dots,C_k$ is a maximal chain of cliques and $uv \notin E(\capc)$ we further know that $C_{uv} \notin \capc$. Since $\capc$ is a maximal chain, $C_u <_{{\cal MA}(P_\sigma)} C_{uv} <_{{\cal MA}(P_\sigma)} C_v$ and $C_{uv} \notin \capc$, we know that there exists a maximal clique $D_1$ in $\capc$ such that $C_{u} <_{{\cal MA}(P_\sigma)} D_1 <_{{\cal MA}(P_\sigma)} C_v$ and $D_1$ covers $C_u$ otherwise we contradict the maximality of the chain. We have three cases: $C_{uv} <_{{\cal MA}(P_\sigma)} D_1$;   $D_1 <_{{\cal MA}(P_\sigma)} C_{uv}$;  $D_1 \parallel_{{\cal MA}(P_\sigma)} C_{uv}$. In the first case, this contradicts the assumption that $D_1$ covers $C_u$. So this case cannot happen. In the second case, we have $u \in C_u, C_{uv}$ and $u \notin D_1$ which contradicts proposition \ref{consecutivity} (consecutiveness condition) and so this case cannot happen. So we are left with the last case. Since $D_1$ covers $C_u$ and  $D_1 \parallel_{{\cal MA}(P_\sigma)} C_{uv}$ we now show that $D_1 \wedge_{{\cal MA}(P_\sigma)} C_{uv}=C_u$. Assume it is not the case; then we would have $C_u<_{{\cal MA}(P_\sigma)} (D_1 \wedge_{{\cal MA}(P_\sigma)} C_{uv}) <_{{\cal MA}(P_\sigma)} D_1$ which contradicts that $D_1$ covers $C_u$. So we are in the situation described in Figure \ref{fig:situation1}.

\begin{figure}[ht]
  \centering
  \includegraphics[scale=.8]{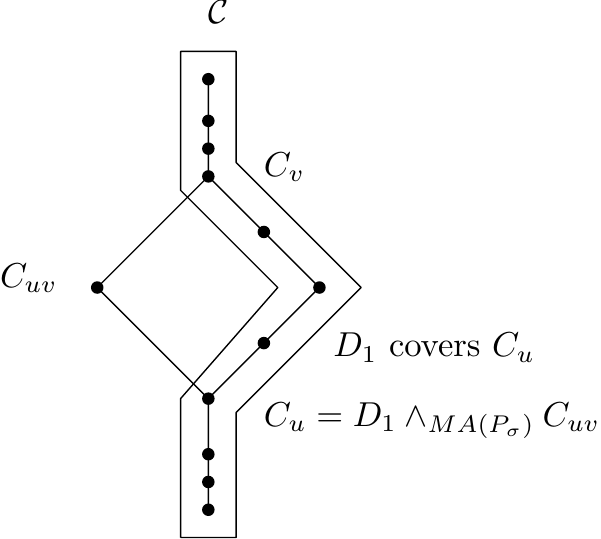}
\begin{center}
\caption{\label{fig:situation1}}
\end{center}
\end{figure}  
	
Since we chose $C_v$ as the first maximal clique in $\capc$ that contains $v$, $v$ is not universal to $D_1$ and let $x$ be a vertex of maximum last value among $D_1 - N(v)$. So $last[x]<first[v]$ Let us show that $x$ belongs to $C_u$. Assume that $x$ belongs to $C_{uv} \vee_{{\cal MA}(P_\sigma)} D_1$, then $v \notin C_{uv} \vee_{{\cal MA}(P_\sigma)} D_1$ and using proposition \ref{supinf} we deduce that $v \in C_{uv} \wedge_{{\cal MA}(P_\sigma)} D_1$. But now $v$ belongs to $C_u$ which contradicts that $uv$ does not belong to $G_\capc$. Thus $x$ is a vertex such that $first[x] \leq last[u] < last[x]<first[v]$. 

	Let $C_x=C_{last[x]}$. Since $\capc$ is an interval graph and $v \notin N(x)$, we see that $C_x  <_{{\cal MA}(P_\sigma)} C_v$. Using the same argument as in the case of $D_1$,  we also see that $C_x \parallel_{{\cal MA}(P_\sigma)} C_{uv}$. From the lattice definition we have that $(C_x \vee_{{\cal MA}(P_\sigma)} C_{uv}) \leq_{MA(P_\sigma)} C_v$ and $C_u \leq_{{\cal MA}(P_\sigma)} (C_x \wedge_{{\cal MA}(P_\sigma)} C_{uv})$. Using proposition \ref{consecutivity} (consecutiveness condition) on $C_{uv} \leq_{{\cal MA}(P_\sigma)} (C_x \vee_{{\cal MA}(P_\sigma)} C_{uv}) \leq_{{\cal MA}(P_\sigma)} C_v $ we see that $v \in (C_x \vee_{{\cal MA}(P_\sigma)} C_{uv})$. Using proposition \ref{consecutivity} (consecutiveness condition) on $C_{u} \leq_{{\cal MA}(P_\sigma)} (C_x \wedge_{{\cal MA}(P_\sigma)} C_{uv}) \leq_{{\cal MA}(P_\sigma)} C_{uv} $ we see that $u \in (C_x \wedge_{{\cal MA}(P_\sigma)} C_{uv})$.
    Since $u \notin C_x$, $u$ is not universal to $C_x$ and let $y$ be a vertex of maximum last value among $C_x - N(u)$. So $last[u]<first[y] \leq last[x]$ and $last[x] \leq last[y]$. Since $u \in (C_x \wedge_{{\cal MA}(P_\sigma)} C_{uv})$ and $y \notin N(u)$, using proposition \ref{supinf} we deduce that $y \in C_x \wedge_{{\cal MA}(P_\sigma)} C_{uv}$ and so $y \in N(v)$. So we have $yv \in E$.
	\end{proof}

    We now carefully choose an edge $uv \in S$ and show that we can find an induced path of length at least 3 in $H$ from $u$ to $v$. Let $uv$ be an edge of $S$ such that $last[u] < first[v]$ and $\nexists x, y \in S$, $last[x] < first[y]$ and (($last[u]<last[x]$ and $first[y]\leq first[v]$) or ($last[u]\leq last[x]$ and $first[y]< first[v]$)). 

    Using the previous claim on $u$ and $v$, we know that there exists $x_1$ and $y_1$ such that $first[x_1]\leq last[u] < first[y_1] \leq  last[x_1] < first[v]$, $last[x_1] \leq last[y_1]$ and $y_1v \in E$. We choose $x_1$ and $y_1$ to be the vertices of maximum $last$ values among the ones satisfying the conditions. By our choice of $uv$ we know that $x_1v \notin E(H)$ and $uy_1 \notin E(H)$. Now we have two cases: either $first[v] \leq last[y_1]$ or $last[y_1]<first[v]$. In the first case, we have that $u,x_1,y_1,v$ is an induced path of length 3 in $H$, and so an induced $C_4$, which contradicts $H$ being a chordal graph. In the second case, we apply the previous claim on $y_1,v$ and deduce that there exists $x_2$ and $y_2$ such that $first[x_2]\leq last[y_1] < first[y_2] \leq  last[x_2] < first[v]$, $last[x_2] \leq last[y_2]$ and $y_2v \in E$. We choose $x_2$ and $y_2$ to be the vertices of maximum $last$ values among the ones satisfying the conditions. By our choice of $uv$ we know that $x_2v \notin E(H)$, $uy_2 \notin E(H)$ and $y_1,y_2 \notin E(H)$.  Since we chose $x_1$ to be the vertex of maximum $last$ value we know that $x_2u \notin E(H)$. Since we chose $y_1$ to be a vertex of maximum $last$ value we know that $x_1x_2 \notin E$. Now we again have two cases: either $first[v] \leq last[y_1]$ or $last[y_1]<first[v]$. In the first case, $u,x_1,y_1,x_2,y_2,v$ is an induced path of length 5 in $H$, and so there is an induced $C_6$, which contradicts $H$ being chordal. In the second case, we do the same argument again. By continuing in this fashion, we always find an induced path from $u$ to $v$ of length at least 3 in $H$. Therefore $H$ cannot be chordal.
    
\end{proof}

\begin{figure}[ht]
  \centering
  \includegraphics[scale=1]{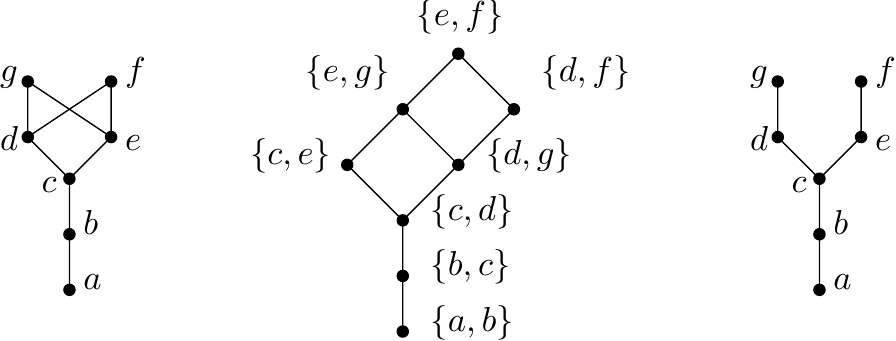}
\begin{center}
\caption{\small{From left to right: a cocomparability graph, along with one of its lattices and a maximal chordal subgraph that is not an interval graph since it contains an asteroidal triple $(a, f, g)$.\label{fig:chsub}}}
\end{center}
\end{figure}  
The statement of theorem \ref{interval-chordal} begs the question of whether 
 all maximal chordal subgraphs of a cocomparability graph are interval subgraphs.
As shown in Figure \ref{fig:chsub} this is not the case.
This naturally leads to the question:   what is the complexity to compute a maximum interval subgraph (i.e., having a maximum number of edges) of a cocomparability graph? Unfortunately it has been shown in \cite{JD} that it is NP-hard.

It is interesting to compare theorem \ref{interval-chordal}  to a result implicit in \cite{HM91} but stated in \cite{Parra} that says the following: every minimal triangulation (or chordalization) of a cocomparability graph is an interval graph. As a corollary,  treewidth and pathwidth are equal for cocomparability graphs.

\section{Algorithmic aspects}\label{ALG}

The problem of finding a maximal chordal subgraph of an arbitrary graph has been studied in \cite{DS88} and an algorithm with complexity $O(n m)$ has been proved. In this section, using a new graph search we will improve this  to $O(n + mlogn)$ for cocomparability graphs.

\subsection{Graph searches and cocomparability graphs}\label{sec:graphsearches}

In the introduction we presented two problems on cocomparability graphs solvable by graph searching where these algorithms 
are very similar to
a corresponding algorithm on interval graphs.  In subsection \ref{list} we present other problems where this ``lifting'' technique
provides new easily implementable cocomparability graph algorithms.  All of the algorithms that we mention use a technique
called the ``$+$ tie-break rule'' in which a total ordering $\tau$ of $V(G)$ is used to break ties in a particular graph search $\cal{S}$.  
In particular,
the next chosen vertex in $\cal{S}$ will be the \textit{rightmost} tied vertex in $\tau$.  Such a tie-breaking search will be denoted 
$\cal{S}$$^+(\tau)$.  
%
%
%
%
%
 Many of these examples use that fact that some searches (most
notably LDFS) when applied  as a ``$+$-sweep'' to a cocomp ordering produce a vertex ordering that is also a cocomp ordering.
In fact, in \cite{BigArt} there is a characterization of the graph searches that have this property of preserving a cocomp ordering.
Given a cocomparability graph $G$, computing a cocomp ordering can be done in linear time,~\cite{MS99}. This algorithm, however, is quite
involved and other algorithms with a running time in $O(n+m\log(n))$ are easier to implement~\cite{MS99, HMPV00}. It should be noticed that up to now, it is not known if one can check
if an ordering is a cocomp ordering in less than boolean matrix
multiplication time.

%

\subsection{Other examples of graph searches on cocomparability graphs}\label{list}

Following the two examples presented in the introduction we now present three other examples 
of search based algorithms for other problems on cocomparability graphs: 

\begin{itemize}
\item Let $x$ be the last vertex of an arbitrary LBFS of cocomp graph
  $G$ and let $y$ be the last vertex of an arbitrary LBFS starting at
  $x$.  Then $\{x,y\}$ forms a {\it dominating pair} in the sense that
  for all $[x, y]$ paths in $G$, every vertex of $G$ is either on the
  path or has a neighbor on the path~\cite{COS299}.\footnote{In fact
    this result was proved for the larger family of {\it asteroidal
      triple-free (AT-free) graphs} and was the first use of LBFS outside the chordal graph family.  }
\item Let $\sigma$ be an LDFS cocomp ordering of graph $G$.  Then a
  simple dynamic programming algorithm for finding a longest path in
  an interval graph also solves the longest path problem on cocomparability
  graphs where $\sigma$ is part of the input to the
  algorithm~\cite{MC}.
\item Let $\sigma$ be an LDFS cocomp ordering of graph $G$.  Then a
  simple greedy algorithm for finding the maximum independent set (MIS) in
  an interval graph also solves the problem on cocomparability
  graphs where $\sigma$ is part of the input to the algorithm \cite{BigArt}.  It is well known that for
  any graph $G= (V(G), E(G))$ with MIS $X \subseteq V(G)$, the set $Y = V(G) \setminus X$ forms
  a minimum cardinality vertex cover (i.e., every edge in $E(G)$ has at least one endpoint in $Y$).
  The MIS algorithm in \cite{BigArt} certifies the constructed MIS by constructing a clique cover (i.e., 
  a set of cliques such that each vertex belongs to exactly one clique in the set) of the same
  cardinality as the MIS.  Recall that cocomparability graphs are perfect.
  
\end{itemize}

 The last two algorithms in the list as well as the two in the introduction suggest the existence of an interesting relationship between interval and cocomparability graphs. We believe that the basis of this relationship is the lattice ${\cal MA}(P)$, which characterizes cocomparability graphs and  shows that a cocomparability graph can be seen as a composition of interval graphs (i.e., the maximal chains of cliques of ${\cal MA}(P)$). 

\subsection{Computing interval subgraphs of a cocomparability graph}\label{sec:intervalstruct}

 In this section, we develop an algorithm that computes a maximal chain of the lattice ${\cal MA}(P_\sigma)$ and show that it forms a maximal interval and chordal subgraph. The problem of finding a maximal interval subgraph is the dual of the problem of finding a minimal interval completion (see \cite{Todinca1,Todinca2}). The algorithm that we are going to present uses a new graph search that we call LocalMNS, since it shares a lot of similarities with MNS. 
 

This algorithm also gives us a way to compute a minimal interval extension of a partial order. An interval extension of a partial order is an extension that is also an interval order. In \cite{HMPR91,HMPR92}, it has been proved that the maximal chains of ${\cal MA}(P)$ are in a one-to-one correspondence with the minimal interval extensions. Therefore, our algorithm also allows us to compute a minimal interval extension of a partial order in $O(n+mlogn)$ time.

This section is organized as follows. First we present a greedy algorithm, called Chainclique with input a total ordering of
an arbitrary graph's vertex set, that computes an interval subgraph. This idea has already been described in \cite{DOS09} for extracting the maximal cliques of an interval graph from an interval ordering. Here we generalize it in order to accept  as input any  graph and any ordering.  In subsection \ref{sec:mininterext}, we present a new graph search named LocalMNS. We will also prove that applying algorithm Chainclique on a LocalMNS cocomp ordering produces a maximal chain of ${\cal MA}(P)$. In subsection \ref{sec:maximalinter}, we have shown that such a maximal chain of ${\cal MA}(P)$ forms a maximal interval and chordal subgraph. 



\begin{definition}
	Let $\Gr$ be a graph and $\sigma$ an ordering of $\V$.

	A graph $H=(\V,E(H))$ with $E(H) \subseteq \E$ is a \textit{$\sigma$-maximal interval subgraph} for the ordering $\sigma$ if and only if $\sigma$ is an interval ordering for the graph $H$ and $\forall S \subseteq \E-E(H)$, $S\neq \emptyset$, $\sigma$ is not an interval ordering for the graph $H'=(\V,E(H) \cup S)$.
\end{definition}

\vspace{0.5cm}

\begin{algorithm}[ht] 
 \KwData{$\Gr $ and a vertex ordering $\sigma$}
 \KwResult{a chain of cliques $C_1$,...,$C_j$}
 $j \leftarrow 0$\;
 $i \leftarrow 1$\;
 $C_0 \leftarrow \emptyset$\;
 \While{$i \leq |V|$}{
	 $j \leftarrow j+1$ \hspace{.6in} \%\{Starting a new clique\}\%\;
	 $C_j \leftarrow \{\sigma(i)\} \cup (N(\sigma(i)) \cap C_{j-1}$)\;
	 $i \leftarrow i+1$\;

	 \While{$i \leq |V|$ and $\sigma(i)$ is universal to $C_j$ }{
		 $C_j \leftarrow C_j \cup \{\sigma(i)\}$ \hspace{.6in} \%\{Augmenting the clique\}\%\;
		 $i \leftarrow i+1$\;
	 }
 }
 Output $C_1,\dots,C_j$\;
\caption{Chainclique($G, \sigma$) \label{CUI}}
\end{algorithm}
\vspace{0.5cm}

As we will prove, Chainclique($G, \sigma$) computes a $\sigma$-maximal interval subgraph for an arbitrary given graph $G$. To this end, Chainclique($G, \sigma$) computes a sequence of cliques that respects the consecutiveness condition. Chainclique($G, \sigma$) tries to increase the current clique and when it cannot, it creates a new clique and sets it to be the new current clique. Another way to see it is that Chainclique($G, \sigma$) discards all the edges $xz \in \E$ such that $\exists y$, $x <_\sigma y <_\sigma z$ and $xy \notin \E$.



It should be noticed that the cliques produced by Chainclique($G, \sigma$) are not necessarily maximal ones, for example take a $P_3$ on the 3 vertices $u,~v,~w$  with the edges $uv$ and $vw$.  Chainclique($P_3, \sigma$) with $\sigma= u,~w,~v$, produces the cliques: $\{u\}, \{w,~v\}$. It should also be noted that the algorithm works on an arbitrary graph and with an arbitrary ordering. For an example, let us consider the graph $H$ of Figure \ref{fig:exemplech} and the ordering $\tau= v_1,~v_2,~v_3,~v_4,~v_5,~v_6,~v_7,~v_8,~v_9,~v_{10},~v_{11},~v_{12},~v_{13},~v_{14},~v_{15}$. Chainclique($H, \tau$) outputs $\{v_1,~v_2,~v_3\},~\{v_2,~v_3,~v_4\},~\{v_3,~v_4,~v_5\},~\{v_5,~v_6\},$ $~\{v_6,~v_7,~v_8\},$ $~\{v_7,~v_8,~v_9\},$ $~\{v_8,~v_9,~v_{10}\}$, $\{v_{11},~v_{12},~v_{13}\}$, $\{v_{11},~v_{13},~v_{14}\}$, $\{v_{13},~v_{14},~v_{15}\}$.  For the graph $H$ presented in Figure \ref{fig:exemplech}, 
$E(H) - E(\capc_H) = \{\{v_7,~v_{11}\}, ~\{v_4,~v_{12}\}\}$.

\begin{figure}[ht]
  \centering
  \includegraphics[scale=1.5]{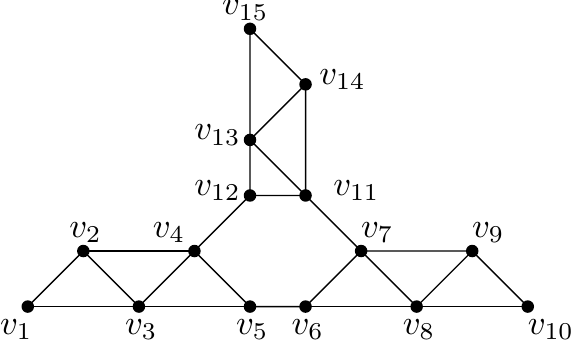}
\begin{center}
\caption{ The graph $H$\label{fig:exemplech}}
\end{center}
\end{figure}  

Now let us prove that Chainclique($G, \sigma$) allows us to obtain a maximal interval subgraph for an ordering $\sigma$. 
The proof is organized as follows. In the first proposition we prove that Chainclique($G, \sigma$) outputs a sequence of cliques that respects the consecutiveness property. In the second proposition, we prove that the ordering given to Chainclique is an interval ordering for the sequence of cliques. In the last proposition we prove that the graph formed by the sequence is a maximal interval subgraph for the ordering.

\begin{property}\label{propconse}
	For a graph $G$ and an ordering $\sigma$, Chainclique($G, \sigma$) outputs a sequence of cliques $\capc=C_1,\dots,C_k$ such that for every $C_{e}, C_{f}, C_{g}$, $1\leq e \leq f \leq g \leq k$, $C_{e}\cap C_{g}\subseteq C_{f}$.

\end{property}
\begin{proof}
	We do the proof by induction on the cliques of $\capc$ and the induction hypothesis is that at each step $j$ if  $x \in C_{j-1}$-$C_{j}$ then $x \notin C_{j'}$, $j' \ge j$.
Since $C_0=\emptyset$, the hypothesis is true for the initial case, $j=1$.


    Assume that the hypothesis is true for the first $j\geq 1$ cliques. When we start to build the clique $C_{j+1}$, we add a vertex that has not been considered before and its neighborhood in $C_j$. By doing so, we cannot add a vertex $x$ to $C_{j+1}$ such that $x \in C_i$-$C_j$ and $i<j$. When we increase the clique, we only add vertices that have not been considered before and so we cannot add a vertex $x$ such that $x \in C_i$-$C_j$ and $i<j$, in $C_{j+1}$. Therefore the induction hypothesis is also verified at step $j+1$.
\end{proof}

Therefore using the characterization of interval graphs of corollary \ref{intervalcharac}, Chainclique($G, \sigma$) outputs a sequence of cliques that defines an interval subgraph.

\begin{property}\label{interorder}
	For a graph $G$ and an ordering $\sigma$, Chainclique($G, \sigma$) outputs a sequence of cliques $\capc=C_1,\dots,C_k$ such that $\sigma$ is an interval ordering for $G_\capc$ and $\forall x \in C_i$-$C_j$, $\forall y \in C_j$-$C_i$, $i<j$, $x <_\sigma y$.
\end{property}

\begin{proof}
    Assume for contradiction that $\sigma$ is not an interval ordering for $G_\capc$. So there exists $u<_\sigma v<_\sigma w$ such that $uv \notin E(\capc)$ and $uw \in E(\capc)$. Let $C_u$ be the first clique in which $u$ appears, $C_v$ be the first clique in which $v$ appears and $C_{uw}$ the first clique that contains both $u$ and $w$. Because Chainclique($G, \sigma$) considers the vertices in the order they appear in $\sigma$, the clique $C_u$ must appear in $\capc$ before the clique $C_v$. Using the same argument the clique $C_v$ must appear before the clique $C_{uw}$. But now $C_u$, $C_v$, $C_{uw}$ contradict  proposition \ref{propconse}, since $u \notin C_v$.

    Now assume for contradiction that $\exists x \in C_i$-$C_j$, $\exists y \in C_j$-$C_i$, $i<j$, $y <_\sigma x$. Now the vertices are considered by Chainclique($G, \sigma$) in the order they appear in $\sigma$. Since $y <_\sigma x$, let $C_g$ be the first clique in which $y$ appears. We see that $g \leq i$. Since $y$ belongs to $C_g$ and $C_j$, using proposition \ref{propconse} we know that $y \in C_i$. Therefore $y \notin C_j$-$C_i$, which contradicts our choice of $y$. Thus $\forall x \in C_i$-$C_j$, $\forall y \in C_j$-$C_i$, $i<j$, $x <_\sigma y$.
\end{proof}

We are ready to prove that the graph formed by the sequence is a $\sigma$-maximal interval subgraph.

\begin{property}
	For a graph G and an ordering $\sigma$, Chainclique($G, \sigma$) outputs a sequence of cliques $\capc=C_1,\dots,C_k$ that induces a maximal interval subgraph for the ordering $\sigma$.
\end{property}

\begin{proof}
    Assume for contradiction that $\capc=C_1,\dots,C_k$ does not form a $\sigma$-maximal interval subgraph. Therefore there exists a non empty set of edges $S$ such that $\sigma$ is an interval ordering for the graph $H=(\V,E(\capc)\cup S)$. Let $uv$ be an edge of $S$ and assume without loss of generality that $u<_\sigma v$. Let $C_i$ be the last clique of $\capc$ containing $u$ and consider $w$ the first vertex of $C_{i+1}$, as chosen by Chainclique; clearly $uw \notin E$ and thus $w \neq v$. Now $u <_\sigma w <_\sigma v$ contradicts $\sigma$ being an interval ordering for the graph $H$.
\end{proof}

\begin{property}
    Chainclique$(G, \sigma)$ has complexity $O(n+m)$. 
\end{property}

\begin{proof}
    All the tests can be performed by visiting once the neighborhood of a vertex. Since the sequence of cliques forms a subgraph of $G$, its size is bounded by $m$. Therefore, Chainclique$(G, \sigma)$ has complexity $O(n+m)$.
   \end{proof}

\subsection{Computing a maximal chain in the lattice}\label{sec:mininterext}

In this subsection, we introduce a new graph search that will be used as a preprocessing step in the computation of a maximal chain of ${\cal MA}(P)$. This graph search will be called LocalMNS and when we use Chainclique$(G, \sigma)$ on a LocalMNS cocomp ordering $\sigma$ we will obtain a maximal chain of ${\cal MA}(P)$.

First, we start by looking at the behavior of Chainclique$(G, \sigma)$ in which $\sigma$ is a LBFS or a LDFS ordering. Let us consider the graph in Figure \ref{fig:exemple2}. Applying the algorithm Chainclique$(G, \sigma)$ on the LDFS ordering $\sigma=1, 3, 2, 4, 6, 5$, we get the chain of cliques $\{1,2,3\}$, $\{1,2,4\}$ and $\{4,5,6\}$ which is not a maximal chain of ${\cal MA}(P)$. A similar result holds using the LBFS ordering $\tau=2, 3, 1, 4, 6, 5$. Thus, LBFS and LDFS do not help us find a maximal chain of cliques of ${\cal MA}(P)$ using Chainclique$(G, \sigma)$. This motivates the introduction of LocalMNS (algorithm \ref{localMNS}).

\begin{figure}[!ht]
  \centering
  \includegraphics[scale=.7]{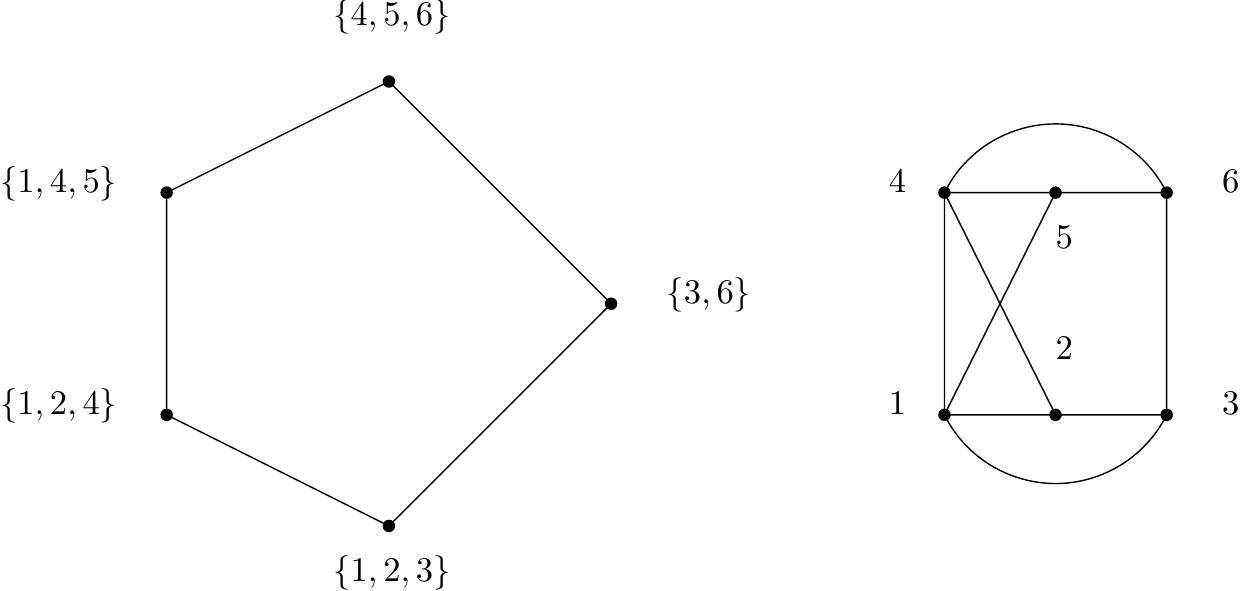}
\begin{center}
\caption{${\cal MA}(P)$ and the corresponding cocomparability graph G\label{fig:exemple2}}
\end{center}
\end{figure}  

\begin{algorithm}[!h] 
 \KwData{$G=(V,E) $}
 \KwResult{a total ordering  $\sigma$ such that  $\sigma(i)$ is the i'th visited vertex }
 $D_1 \leftarrow \emptyset$\; 
 $V' \leftarrow V$ \hspace{.6in}\%\{$V'$ is the set of unchosen vertices\}\%\;
 $X \leftarrow \emptyset$ \hspace{.6in}\%\{$X$ is the set of chosen vertices\}\%\;\;
 \For{$i= 1$ to $|V|$}{
   $v$ is chosen as a vertex from $V'$ with maximal neighborhood in $D_{i}$\;
   $\sigma(i) \leftarrow v$\;
   $V' \leftarrow V'-\{v\}$\; 
   $X \leftarrow X \cup \{v\}$\;
   $D_{i+1} \leftarrow \{v\} \cup (N(v) \cap D_{i})$; \hspace{0.5cm}\% Note: $x \in D_i$-$D_{i+1}$ $\rightarrow x\notin D_j, j>i$\%\;
}
\caption{LocalMNS \label{localMNS}}
\end{algorithm}
\vspace{0.5cm}

This algorithm is very similar to the standard Maximal Neighborhood Search (MNS) algorithm.  The only difference is in LocalMNS we are considering the neighborhood of the unvisited vertices only in $D_i$, which can be a strict subset of $X$ (the visited vertices) and in the case of MNS we are considering the neighborhood in $X$. This is the reason for the name LocalMNS. Let us look at the behavior of $LocalMNS^+$ on the example of Figure \ref{fig:exemple2}. Let $\tau=5,6,4,2,3,1$ be a cocomp ordering. $\sigma =LocalMNS^+(G,\tau)=1,3,2,4,5,6$ is a cocomp ordering and $Chainclique(G, \sigma)$ computes the maximal chain $\{1,2,3\}$, $\{1,2,4\}$, $\{1,4,5\}$ and $\{4,5,6\}$.  

\begin{property}
LocalMNS can be implemented in linear time.
\end{property}
\begin{proof}
It is well known that MNS can be implemented in linear time via MCS. So we will use a LocalMCS to compute LocalMNS. LocalMCS works the same ways as LocalMNS except that at each step $i$, instead of choosing a vertex of maximal neighborhood in $D_i$, we choose a vertex with maximum degree in $D_i$. 

To implement LocalMCS we use  a partition refinement technique. We use an ordered partition in which each part contains the vertices of $V'$ having a given degree in $D_i$.
This structure can be easily maintained  when $D_{i+1}$ is formed from $D_i$.  If a vertex $x \in D_i$ does not appear in
$D_{i+1}$, then $x \notin D_j, j>i$ and thus, during the execution of LocalMCS, we visit the neighborhood of each vertex at most 2 times.
Therefore LocalMCS is in $O(n+m)$.

\end{proof}

\begin{property}\label{complMNS}
    LocalMNS$^+(G, \sigma)$ can be implemented in  $O(n+mlogn)$. 
\end{property}
\begin{proof}
As in the previous proposition  $LocalMNS^+$ can be implemented via  $LocalMCS^+$. 
We will also use an ordered partition in which each part contains the vertices of $V'$ having a given degree in $D_i$. But each part has to be ordered with respect to $\sigma$. This is the bottleneck of this algorithm. To handle this difficulty each part will be represented by a tree data structure. This leads to an algorithm in O(n+mlogn).

\end{proof}
\vspace{0.5cm}

We now prove that Chainclique$(G, \sigma)$ on a LocalMNS cocomp ordering outputs a maximal chain of cliques. Let $G$ be a cocomparability graph and $\tau$ a cocomp ordering of $G$. The proof is organized as follows. We first show that $\sigma=LocalMNS^+(G,\tau)$ is a cocomp ordering. Then we describe the structure of a maximal chain of ${\cal MA}(P_\sigma)$ and its relation to $P_\sigma$. Finally we prove that $Chainclique(G,\sigma)$ outputs a maximal chain of ${\cal MA}(P_\sigma)$.

%

\begin{lemma}\label{localMNScocomp}
    If G is a cocomparability graph, then $\tau$ is a cocomp ordering if and only if $\sigma=LocalMNS^+$$(G, \tau)$ is a cocomp ordering.
\end{lemma}
\begin{proof}

Note that this lemma can be stated as a corollary of the characterization in \cite{BigArt} of the searches that preserve being a cocomp
ordering.  Instead we give a direct proof.

First we show that $\sigma$ and $\tau$ satisfy the ``flipping property'' insofar as two nonadjacent vertices $u, v$ must be in different
relative orders in the two searches.  To prove this, we assume that $u <_{\sigma} v$ and $u <_{\tau} v$ where, without loss of
generality, $u$ is the leftmost vertex in $\sigma$ that has such a non flipping non neighbor $v$.  Now, because of the ``$+$'' rule
in order for $u <_{\sigma} v$ at the time $u$ was selected by $\sigma$, there must exist a previously visited vertex $w$ in $\sigma$
such that $uw \in E(G)$, $vw \notin E(G)$.  Note that $w <_{\sigma} u <_{\sigma} v$.  By the choice of $u$ in $\sigma$, we see
that $v <_{\tau} w$ and thus there is an umbrella $u <_{\tau} v <_{\tau}  w$ in $\tau$, contradicting $\tau$ being a cocomp ordering.  

Now assume that $\tau$ is a cocomp ordering but $\sigma$ is not.  Let $a <_{\sigma} b <_{\sigma}  c$ be an umbrella in $\sigma$ where $ac \in E(G), 
ab, bc \notin E(G)$.  By the ``flipping property'', $b <_{\tau} a$ and $c <_{\tau} b$ thereby showing that $c <_{\tau} b <_{\tau}  a$ forms
an umbrella in $\tau$ contradicting $\tau$ being a cocomp ordering.

The rest of the proof follows immediately.
\end{proof}

Let us introduce some terminology to help us describe the behavior of Chainclique on an ordering $\sigma$. Let $j_i$ be the first value of $j$ such that  $\sigma (i)$ belongs to $C_j$  (i.e., $C_{j_i}$ is the leftmost clique containing $\sigma (i)$. Let $C_{j}^1,\dots ,C_{j}^{l_j}$ be the sequence to build the clique $C_{j}$. Let $p_i$ be the first value of p such that $\sigma(i)$ belongs to the clique $C_{j_i}^p$.

We are ready to prove that Chainclique$(G, \sigma)$ where $\sigma$ is a LocalMNS cocomp ordering  outputs a maximal chain of ${\cal MA}(P_\sigma)$. The proof is organized as follows. The next claim links Chainclique$(G, \sigma)$ and LocalMNS. In proposition \ref{maxclique}, we prove that the cliques output by Chainclique$(G, \sigma)$ where $\sigma$ is a LocalMNS cocomp ordering are maximal cliques. In theorem \ref{th:maxchain}, we prove that the chain is a maximal chain of  ${\cal MA}(P_\sigma)$.

\begin{claim}\label{setdi}
    Let G be a cocomparability graph and  let $\tau$ be a cocomp ordering.
    
    If  $\sigma=LocalMNS^+(G, \tau)$ and $Chainclique(G, \sigma) = C_1,\dots,C_k$ then for all values of $i$ the set $D_{i+1}$ computed by LocalMNS equals the set $C_{j_i}^{p_i}$ computed by $Chainclique(G, \sigma)$.

\end{claim}

\begin{proof}
    The proof is by induction. The inductive hypothesis is $H_i$:  for all $i \ge 1$, $D_{i+1}=C_{j_i}^{p_i}$. 
    
    
Since $D_{2}=\{\sigma(1)\}$ and $C_{1}^{1}=\{\sigma(1)\}$, $H_1$ is trivially true. 
 Assume the hypothesis is true for the first i-1 vertices with $i>1$. We have two cases: $\sigma(i)$ is complete to $D_i$ or not. In the first case, LocalMNS increases the set $D_i$ by adding $\sigma(i)$ and Chainclique$(G, \sigma)$ will increase the clique $C_{j_{i-1}}^{p_{i-1}}$ by adding $\sigma(i)$. Therefore using $H_i$, we deduce that $D_{i+1}=C_{j_i}^{p_i}$. In the second case, LocalMNS sets $D_{i+1}=D_i \cap N(\sigma(i))$. In the same way, $Chainclique(G, \sigma)$ creates a new clique $C_{j_i}^1= C_{j_{i-1}}^{p_{i-1}}  \cap N(\sigma(i))$. Therefore using $H_i$, we deduce that $D_{i+1}=C_{j_i}^{p_i}$.

\end{proof}

\begin{property}\label{maxclique}
    Let G be a cocomparability graph and let $\tau$ be a cocomp ordering.
    
    If  $\sigma=LocalMNS^+(G, \tau)$ then $Chainclique(G, \sigma)=C_1,\dots,C_k$ is a chain of maximal cliques of ${\cal MA}(P_\sigma)$ such that $C_1 <_{{\cal MA}(P_\sigma)} C_2 <_{{\cal MA}(P_\sigma)} \dots <_{{\cal MA}(P_\sigma)} C_k$.
\end{property}
\begin{proof}
    We start by proving that $\capc=C_1,\dots,C_k$ are all maximal cliques of $G$ and then $C_1 <_{{\cal MA}(P_\sigma)} C_2 <_{{\cal MA}(P_\sigma)} \dots <_{{\cal MA}(P_\sigma)} C_k$.
    
    Assume for contradiction that some cliques are not maximal and let $C_g$ be the first clique of the chain which is not a maximal clique of $G$. Let $w$ be a vertex complete to $C_g$ but $w \notin C_g$ and let $w$ be the rightmost such vertex in $\sigma$. Let $v$ be the first vertex in $\sigma$ of $C_{g}- C_{g-1}$. We have two cases: either $v <_\sigma w$ or $w <_\sigma v$.

    In the first case, let $u=\sigma(i)$ be first vertex of $C_{g+1}-C_{g}$. Since $w \notin C_g$, we must have $u<_\sigma w$. Using claim \ref{setdi}, we see that in LocalMNS$^+$ at step $i-1$, 
    $D_i=C_g$. But now at the time $u$ was chosen, $w$ is complete to $D_i$ but $u$ is not, thereby, contradicting LocalMNS$^+$ choosing $u$.
    
    In the second case, let $C_h$ be the last clique in the chain with $w$. We must have that $h<g$ since $w <_\sigma v$. Since $w$ is a neighbor of $v$, $w$ is not in $C_{g-1}$ (otherwise $w \in C_g$) and we have that $h+1<g$. So now let us consider the first vertex $x$ of $C_{h+1}-C_h$ in the ordering $\sigma$. Since $w$ does not belong to $C_{h+1}$, we know that $wx \notin E$ and since $w$ is universal to $C_g$ we have that $x \notin C_g$. Because $w$ is the rightmost complete vertex to $C_g$, $x$ must not be adjacent to some vertex $y$ of $C_g$. Now either $x <_\sigma y$ or $y <_\sigma x$. In the first case we know that $w <_\sigma x <_\sigma y$ and $w,x,y$ is an umbrella. Therefore $\sigma$ is not a cocomp ordering, which is a contradiction to lemma \ref{localMNScocomp}. In the second case since $y$ appears before $x$, we have that $first[y] \leq h+1$. But now $y$ belongs to $C_{first[y]}$ and $C_g$ and so using property \ref{propconse}, we know that $y \in C_{h+1}$. This is a contradiction to $xy \notin E$. Thus all the cliques in $\capc$ are maximal cliques of $G$.

    Let us now prove that  $C_1 <_{{\cal MA}(P_\sigma)} C_2 \dots <_{{\cal MA}(P_\sigma)} C_k$. For this purpose, we show that if $i<j$ then $\forall x \in C_i$, $\exists y \in C_j$ such that $xy \notin E$, $x \leq_\sigma y$. Assume first that $x\in C_i\cap C_j$ then we have that $xx \notin E$ and  $x \leq_\sigma x$. Now assume that $x\in C_i- C_j$. Since $C_j$ is a maximal clique of $G$ and $x$ does not belong to $C_j$, there exists $y \in C_j -C_i$ such that $xy \notin E$. We know by proposition \ref{interorder} that $\forall x \in C_i-C_j$, $\forall y \in C_j-C_i$, $i<j$, $x <_\sigma y$. Therefore  $x <_\sigma y$ and so $C_i  <_{{\cal MA}(P_\sigma)} C_j$. 
\end{proof}

Let us now prove that Chainclique forms a maximal chain.

\begin{theorem}\label{th:maxchain}
    For a cocomparability graph G and a cocomp ordering $\tau$ of G, if  $\sigma=LocalMNS^+(G, \tau)$ then $Chainclique(G, \sigma)$ is a maximal chain of maximal cliques of ${\cal MA}(\sigma$).
\end{theorem}

\begin{proof}
    Let $\capc=C_1,\dots,C_k$ be the chain of cliques output by $Chainclique(G, \sigma)$. In proposition \ref{maxclique}, we proved that $\capc$ forms a chain of maximal cliques of ${\cal MA}(P_\sigma)$. Thereby we only have to prove now that the chain is maximal. 
    
    We will show in the next three claims that $C_1$ is the set of sources of $P_\sigma$, that $C_j$ covers $C_{j-1}$ and finally that $C_k$ is the set of sinks of $P_\sigma$. Using proposition \ref{structchain}, we will be able to deduce that $ C_1 <_{{\cal MA}(P_\sigma)} C_2 <_{{\cal MA}(P_\sigma)} \dots <_{{\cal MA}(P_\sigma)} C_k$ is a maximal chain of ${\cal MA}(P_\sigma)$.
    
    \begin{claim}
    $C_1$ is the set of sources of $P_\sigma$.
    \end{claim}
    \begin{proof}
    For the initial case, let $C_S$ be the set of sources of $P_\sigma$. We will show that $C_1=C_S$ by proving that $\sigma$ starts with all the vertices of $C_S$ and only them. Since $\sigma$ is a linear extension of $P_\sigma$, $\sigma$ starts with at least one source. So we suppose without loss of generality that $\sigma$ starts with a set of sources $S \subsetneq C_S$ and $S\neq \emptyset$. Now assume for contradiction that after $S$, we have a vertex $x$ such that $x \notin C_S$. Let $i=\siginv(x)$. All the sources after $x$ are complete to $S$ and so for LocalMNS to choose $x$, $x$ must also be universal to $S$ since at this step $D_i$ is equal to $S$. Now since $x$ does not belong to $C_S$, there is a vertex $v \in C_S$ that is comparable to $x$ and because $C_S$ is the set of sources of $P_\sigma$ and since $\sigma$ is a linear extension of $P_\sigma$, we must have that $v <_\sigma x$. And so $v$ must belong to $S$. But now $x$ cannot be complete to $S$, which is a contradiction. So $\sigma$ starts with all the sources and only them. 
    \end{proof}
    
    \begin{claim}
    $C_{j-1} \prec_{{\cal MA}(P_{\sigma})} C_{j}$ for $1<j\leq k$. 
    \end{claim}
    \begin{proof}
        Assume for contradiction that $C_j$ does not cover $C_{j-1}$ and that $C_j$ is leftmost with this property. So $C_g$ covers $C_{g-1}$ for $1<g\leq j-1$.         
        Let $A$ be a maximal clique of the lattice such that $A$ covers $C_{j-1}$ and $ A <_{{\cal MA}(P_\sigma)} C_{j}$. Using proposition \ref{consecutivity} on $C_{j-1}$, $C_j$ and $A$, we deduce that $C_{j-1} \cap C_j \subsetneq A $.
        We have two cases: either $A \subseteq C_j \cup C_{j-1}$ or $A \not\subset C_j \cup C_{j-1}$.

        In the first case, since $C_{j-1} \cap C_j \subseteq A$, we have $C_{j-1} \cap C_j \subseteq A \cap C_{j-1}$. Assume for contradiction that $C_{j-1} \cap C_j = C_{j-1} \cap A$. Using  $A \subset C_j \cup C_{j-1}$ and $C_{j-1} \cap C_j = C_{j-1} \cap A$ we can deduce that $A \subset C_j$, which contradicts the maximality of $A$. Therefore  $C_{j-1} \cap C_j \subsetneq A \cap C_{j-1}$. Let $v=\sigma(i)$ be the leftmost vertex of $C_j-C_{j-1}$ in $\sigma$. We again have  two cases: either $v \in A$ or $v \notin A$. In the first case, Chainclique set $C_j^1=(N(v) \cap C_{j-1}) \cup \{v\}$. But since  $C_{j-1} \cap C_j \subsetneq A \cap C_{j-1}$ and $v \in A$, we have $C_{j-1} \cap C_j \subsetneq N(v) \cap C_{j-1}$ and so $C_j^1 \not\subset C_j$. This is a contradiction to the behavior of Chainclique. In the second case, using claim \ref{setdi} on $v$ we deduce that the set $D_i$ of LocalMNS equals $C_{j-1}$. But now let $x \in A -C_{j-1}$. Since $C_{j-1} \cap C_j \subsetneq A \cap C_{j-1}$, we have that $N(v) \cap D_i \subsetneq N(x) \cap D_i$. This is a contradiction to the choice of LocalMNS. 

        In the second case, let $x \in A-(C_j \cup C_{j-1})$.  Assume for contradiction that there exists a maximal clique $B$ such that $x \in B$ and $B <_{{\cal MA}(P_\sigma)} C_{j-1}$. Now $x \in B, A$ and $x \notin C_{j-1}$ contradicting proposition \ref{consecutivity}. Therefore $x$ cannot appear in $\sigma$ before the last vertex of $C_{j-1}$. Since $x \notin C_j$, $\exists y \in C_j$ such that $xy \notin E$. Using lemma \ref{defMA} on $x$, $y$ we deduce that $x <_{P_\sigma} y$ and since $\sigma$ is a linear extension of $P_\sigma$ we know that $x <_{\sigma} y$. But now since $x$ appears after the last vertex of $C_{j-1}$, before the last of $C_j$ and $x$ is not complete to $C_j$, Chainclique must build a clique in the sequence between $C_{j-1}$ and $C_j$, which is a contradiction.

    \end{proof}

    \begin{claim}
    $C_k$ is the set of sinks of $P_\sigma$.
    \end{claim}
    \begin{proof}
        For the final case, we show that $C_k$ is the set of sinks of $P_\sigma$. Assume for contradiction that $x$ is a sink of $P_\sigma$ and $x$ does not belong to $C_k$. All the vertices belong to at least one clique of $\capc$ and let $C_g=C_{last[x]}$. Since $x \notin C_k$, we have $g<k$. So $x$ does not belong to $C_{g+1}$. But now let $y$ be the first vertex in $\sigma$ of $C_{g+1} - C_g$. Since $x \notin C_{g+1}$ we have $xy \notin E$ and $x <_\sigma y$. But this contradicts that $x$ is a sink.  
    \end{proof}

\end{proof}

\begin{corollary}
Let $G$ be a cocomparability graph, then a maximal interval subgraph of $G$ can be computed in $O(n+mlog n)$.
\end{corollary}

To finish let us now show that any maximal chain of ${\cal MA}(P_\sigma)$ can be computed by LocalMNS.

\begin{theorem}
    For a cocomparability graph $G$ and a transitive orientation $P$ of $\overline{G}$, every maximal chain of ${\cal MA}(P)$ can be computed using $Chainclique(G, \sigma)$ on some cocomp  LocalMNS ordering $\sigma$.
\end{theorem}

\begin{proof}
    Let $C_1\prec_{{\cal MA}(P)}...\prec_{{\cal MA}(P)}C_k$ be a maximal chain of ${\cal MA}(P)$.
Let us take the ordering $\tau$ as the interval ordering for this maximal chain of cliques. Now using $LocalMNS^+$ on $\tau^{-1}$, we get $\tau$. So $\tau$ is a LocalMNS cocomp ordering and using Chainclique$(G, \tau)$  we get $C_1<_{{\cal MA}(P)}...<_{{\cal MA}(P)}C_k$.
\end{proof}

Using Theorem \ref{interval-chordal} we immediately have:

\begin{corollary}
A maximal  chordal subgraph of a cocomparability graph $G$ can be computed with complexity $O(n+mlog n)$.
\end{corollary}

\begin{proof}

Preuve à déplacer au dessus ? 

    The algorithm consists of finding a cocomp ordering, then performing a $LocalMCS^+$ and then using Chainclique. A cocomp ordering can be found in $(O(n+m)$ \cite{MS99} and Chainclique has complexity $O(n + m)$. So the bottleneck of this algorithm lies in $LocalMCS^+$. Using proposition \ref{complMNS} the full algorithm can be computed in $O(n+mlog n)$ time.

\end{proof}

%

\section{Computing  all simplicial vertices}\label{sec:simplicial}

In order to compute simplicial vertices we need to consider some particular maximal cliques, called \textbf{fully comparable cliques}.

\begin{definition}
 In a lattice $\cal L=(X,\leq_{\cal L})$, an element $e \in X$ is said to be \emph{fully comparable} if and only if for every $u \in X$, either $e \leq_{\cal L} u$ or $u \leq_{\cal L} e$.
\end{definition}

 We now prove that if $\sigma$ is an $MNS$ cocomp ordering of $G$ then all the fully comparable cliques of $MA(P_\sigma)$ belong to the sequence of cliques obtained using \ChaCli($G$,$\sigma$). As we will show, these cliques play a decisive role in the problem of finding the simplicial vertices of $G$.

\begin{theorem}\label{CliqueInter}
Let $G$ be a cocomparability graph and $\sigma$ a $MNS$ cocomp ordering of the vertices of $G$.
If $C_b$ is a maximal clique such that $C_b$ is fully comparable in ${\cal MA}(P_\sigma)$ then $C_b$ is a maximal clique of the chain output by $\ChaCli(G,\sigma)$.
\end{theorem}

\begin{proof}
    Let $C_b$ be a fully comparable maximal clique in ${\cal MA}(P_\sigma)$. We define $V_{C_b}$ to be $\{x \in V | \exists C_x$ such that $x \in C_x$ and $C_x \leq_{{\cal MA}(P_\sigma)} C_b\}$. 

	Let us first prove that the ordering $\sigma$ starts with all the vertices of $V_{C_b}$. For the sake of a contradiction, let's assume that we have in $\sigma$ a vertex $v$ such that $\exists C_v$, $v \in C_v$ and $C_v >_{{\cal MA}(P_\sigma)} C_b$ before a vertex  $x \in V_{C_b}$ and $v$ is the leftmost such vertex in $\sigma$. Let $C_x$ be a maximal clique such that $x \in C_x$ and $C_x \leq_{{\cal MA}(P_\sigma)} C_b$. We have two cases, either $xv \notin E$ or $xv \in E$.  
	
Case 1: $xv \notin E$.	Using lemma \ref{defMA} on $x,$ $v,$ we deduce that $x <_{P_\sigma} v$ and since $\sigma$ is a linear extension of $P_\sigma$, we know that $x <_{\sigma} v$ contradicting our choice of $v$. 

Case 2:  $xv \in E$. We prove that $N(v) \cap V_{C_b} \subset N(x) \cap  V_{C_b}$. Since $xv \in E$  there exists a maximal clique $D$ such that $\{x,v\} \subset D$ and since $v \notin V_{C_b}$ we know $C_b <_{{\cal MA}(P_\sigma)} D$. Using  proposition \ref{consecutivity} on $C_x,$ $C_b,$ $D$ we deduce that $x \in C_b$. Let $u$ be a vertex of $N(v) \cap V_{C_b}$.  Since $uv \in E$ there exists a maximal clique $A$ such that $\{u,v\} \subset A$ and since $v \notin V_{C_b}$ we know $C_b <_{{\cal MA}(P_\sigma)} A$. Let $C_u$ be a maximal clique such that $u \in C_u$ and $C_u \leq_{{\cal MA}(P_\sigma)} C_b$. Using  proposition \ref{consecutivity} on $C_u,$ $C_b,$ $A$ we deduce that $u \in C_b$. Therefore $ux \in E$ and so $N(v) \cap V_{C_b} \subset N(x) \cap  V_{C_b}$. But now at the time when $v$ was chosen, the label of $v$ can only be equal to the label of $x$. Now $C_b$ is a maximal clique and since $v \notin C_b$, there must exist a vertex $w  \in C_b$ such that $wv \notin E$.   Since $C_v >_{{\cal MA}(P_\sigma)} C_b$,
necessarily  $w <_\sigma v$. Since $x \in C_b$, $wx \in E$ and so the label of $x$ is strictly greater than the label of $v$ when $v$ was chosen which is a contradiction to the choice of MNS. Thus $\sigma$ starts with all the vertices of $V_{C_b}.$
	
	Since $\sigma$ starts with all the vertices of $V_{C_b}$, the ordering of the vertices of $V_{C_b}$ induced by $\sigma$ is a MNS cocomp ordering for the graph induced by $V_{C_b}$. Let $P_{C_b}$ be the transitive orientation of the complement of the graph induced by $V_{C_b}$ obtained using $\sigma$. To prove that $C_b$ belongs to the interval graph computed by $\ChaCli$ we will prove that the last clique that ChainClique  computes using the ordering induced by the vertices of $V_{C_b}$ is the set of sinks of $P_{C_b}$, which is equal to $C_b$. Let $C_1,\dots,C_k$ be the chain of cliques that ChainClique  computes using the ordering induced by the vertices of $V_{C_b}$. Assume for contradiction that $x$ is a sink and $x \notin C_k$. Let $C_g$ be the last clique that contains $x$. Since $x \notin C_k$ we have that $g<k$. Let $y$ be the first vertex in $\sigma$ that belongs to $C_{g+1}-C_g$. Since $x \notin C_{g+1}$, we must have that $xy \notin E$ and so $x <_{P_{C_b}} y,$ therefore contradicting the assumption that $x$ is a sink.
\end{proof}

This result does not hold if  $\sigma$ is not a  MNS cocomp ordering. For example just take a $P_3,$ $u,~v,~w$ and the ordering $u < w < v$ (see Figure \ref{fig:p3fully}). The algorithm cannot output $\{u,v\}$ which satisfies the property.

\begin{figure}[ht]
  \centering
  \includegraphics[scale=1]{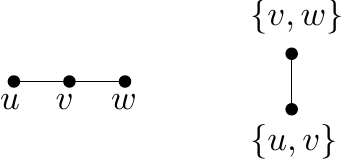}
\begin{center}
\caption{A $P_3$ and its lattice.\label{fig:p3fully}}
\end{center}
\end{figure}  


\begin{theorem}\label{thsimpl}
Let $G$ be  a cocomparability graph and $\sigma$ a cocomp ordering.
If $v$ is a simplicial vertex then there exists a maximal clique $C_v$ such that $v \in C_v$ and $C_v$ is fully comparable in ${\cal MA}(P_\sigma)$.
\end{theorem}

\begin{proof}
Clearly a simplicial vertex belongs to a unique maximal clique. For a simplicial vertex $v$ let us denote by $C_v$ the maximal clique such that $v \in C_v$. Assume that there exists a maximal clique $D$ such that $D \parallel_{{\cal MA}(P_\sigma)} C_v$. Since ${\cal MA}(P_\sigma)$ is a lattice, there exists $D \vee_{{\cal MA}(P_\sigma)} C_v$ and $D \wedge_{{\cal MA}(P_\sigma)} C_v$  two other maximal cliques. Using the definition of  $D \vee_{{\cal MA}(P_\sigma)} C_v$ and $D \wedge_{{\cal MA}(P_\sigma)} C_v$, $v$ belongs to either $D \vee_{{\cal MA}(P_\sigma)} C_v$ or $D \wedge_{{\cal MA}(P_\sigma)} C_v$. Therefore $C_v$ is not the only maximal clique that contains $v$, contradicting $v$ is a simplicial vertex.
\end{proof}

Let us now study an algorithm to find all the simplicial vertices in a cocomparability graph. The correctness of the algorithm mainly relies on theorems \ref{CliqueInter} and \ref{thsimpl}.
To compute the simplicial vertices we need for each vertex a couple of values. We compute the value $first$ and $last$  defined by $first[v]=min\{i | v \in C_i\}$ and $last[v]=max\{i | v \in C_i\}$. The value $forward[v]$ will either be $last[v]$ or the index of the last clique in which $v$  has a neighbor  not in the interval graph and so we define $forward[v]$ to be $max \{\{first[u] | u \in N(v)$ and $first[u]>last[v]\} \cup \{last[v]\}\}$. Note that $forward[v] \geq last[v]$.  

The value $backward[v]$ will either be $first[v]$ or the index of the first clique in which $v$ has a neighbor  not in the interval graph and so we define $backward[v]$ to be $min \{\{last[u] | u \in N(v)$ and $last[u]<first[v]\} \cup \{first[v]\}\}$. Note that $backward[v] \leq first[v]$.

\vspace{0.5cm}
\begin{algorithm}[ht] 
 \KwData{$\Gr$ and an MNS cocomp ordering $\sigma$ of $\V$}
 \KwResult{The simplicial vertices of $G$}
 Compute the sequence $C_1,...,C_k$ using $\ChaCli(G,\sigma)$\;  

 Compute $first$, $last$, $backward$, $forward$ for every vertex of $\V$\;

 $S \leftarrow \emptyset$  \hspace{1.5cm} \%\{The simplicial vertices\}\%\;

 \For{$(i\leftarrow 1$ \KwTo $n)$}{
   \lIf{$forward[\sigma(i)]=backward[\sigma(i)]$}{$S \leftarrow S\cup\{\sigma(i)\}$}
 }

 Output $S$\;
\caption{Simplicial vertices in a cocomparability graph \label{SV}}
\end{algorithm}
\vspace{0.5cm}

\begin{theorem}
Let $G$ be a cocomparability graph and $\sigma$  a MNS cocomp ordering of the vertices of $G$.
Then algorithm \ref{SV} outputs all the simplicial vertices of $G$.
\end{theorem}

\begin{proof}
    Using theorem \ref{CliqueInter} we deduce that all the simplicial vertices and their neighborhood belong to one clique of the sequence output by \ChaCli. So to test if a vertex is a simplicial vertex, we have to check that a vertex belongs to only one clique in the sequence and has no neighbors in the rest of the interval graph. By definition, for any vertex $v$ we have that $last[v] \leq forward[v]$ and $backward[v] \leq first[v]$. So for any simplicial vertex $\sigma(i)$ we have that  $forward[\sigma(i)] = backward[\sigma(i)] = last[\sigma(i)] = first[\sigma(i)]$. 
    If a vertex $\sigma(i)$ is not simplicial then either it belongs to more than one clique in the sequence and so  $last[\sigma(i)]\neq first[\sigma(i)]$ implying $forward[\sigma(i)] \neq backward[\sigma(i)]$ or it has a neighbor outside the sequence and so $last[\sigma(i)] < forward[\sigma(i)]$ or $backward[\sigma(i)] < first[\sigma(i)]$ implying $forward[\sigma(i)] \neq backward[\sigma(i)]$. Therefore our algorithm successfully finds all the simplicial vertices of a cocomparability graph.
\end{proof}

\begin{theorem}
 Simplicial vertices can be computed in linear time on a cocomparability graph, when a cocomp ordering is provided.
 
\end{theorem}

\begin{proof}
To apply algorithm \ref{SV}, we need a MNS cocomp ordering. To obtain such an ordering, if a cocomp ordering $\sigma$ is provided,
we can simply apply LBFS$^+(G, \sigma)$. This can be done in linear time.
$\ChaCli$ has complexity  $O(n+m)$. Enumerating all the cliques takes $O(n+m).$ Computing $first$ and $last$ can be done by enumerating all the cliques of the spanning interval graph. Computing the  $forward$ and  $backward$ functions can be done by enumerating for each vertex its neighborhood after computing $first$ and $last$ and so can be done in $O(n+m)$. Enumerating the vertices can be done in $O(n)$.

\end{proof}

\section{Conclusions and perspectives}\label{CONCL}

In sections \ref{background} and  \ref{list} we presented a number of examples of problems where simple interval graph algorithms
can be ``lifted'' to similar algorithms for cocomparability graphs.  These algorithms are typically based on cocomp orderings 
produced by graph searches, most notably LDFS and LBFS.  The underlying question is whether there is a structural
feature that indicates which problems can be ``lifted" in this way.  In an attempt to answer this question, we have examined the
maximal clique lattice of a cocomparability graph and have presented a characterization theorem of such lattices.  This characterization
has lead to new algorithms for finding a maximal interval subgraph of a cocomparability graph and for finding the set of simplicial
vertices in a cocomparability graph.  Both of these algorithms roughly follow maximal chains of this lattice and in 
the maximal interval subgraph case uses a new graph search, LocalMNS.  In \cite{JD} some other interesting applications of this framework 
have been developed; for example to compute a minimal clique separator decomposition of a cocomparability graph in linear time.

\medskip
\noindent
Our work raises a number of algorithmic questions:

\begin{itemize}
\item Does there exist a LocalMNS that can be implemented in linear time when used as a $+$sweep on cocomp ordering $\sigma$?
Can techniques similar to those used in \cite{KohlerM14} help?
\item Are there other polynomial time solvable interval graph problems that are amenable to the ChainClique approach?

\end{itemize}

\noindent
Similarly our work raises a number of structural questions:

\begin{itemize}
\item
Given our characterization theorem of maximal clique lattices of a cocomparability graph, a natural question is to study the structure 
imposed on cocomparability graphs by restrictions of the lattice structure.
\item  Can anything of interest be found about
the clique structure of  AT-free graphs, the natural generalization of cocomparability graphs?
\item In section \ref{sec:intervalstruct}, we have exhibited some relationships shared by cocomparability graphs and interval graphs and the importance of graph searches in cocomparability graphs. But, we still have not managed to give a full answer to the question of why some interval graph algorithms can be ``lifted'' to work on cocomparability graphs.  Does there exist some generic greedoid structure for  cocomparability graphs that explains why these greedy algorithms work? So far we have no good answer for this question. 
\end{itemize}

\bigskip
\noindent
\textbf{Acknowledgements:}
DGC wishes to thank the Natural Sciences and Engineering Research Council (NSERC) of Canada for financial support of this research.


\begin{thebibliography}{10}

\bibitem{Berhendt}
Behrendt B.
\newblock Maximal antichains in partially ordered sets.
\newblock {\em Ars Combinatoria}, 25C:149--157, 1988.

\bibitem{Birkhoff}
G.~Birkhoff.
\newblock Lattice theory.
\newblock {\em American Mathematical Society}, 25(3), 1967.

\bibitem{BLS99}
Andreas Brandst\"adt, Van~B. Le, and Jeremy~P. Spinrad.
\newblock {\em Graph Classes, a survey}.
\newblock SIAM Monographs on Discrete Mathematics and Applications, 1999.

\bibitem{CLM12}
N.~Caspard, B.~Leclerc, and B.~Monjardet.
\newblock {\em Finite Ordered Sets Concepts, Results and Uses}.
\newblock Cambridge University Press, 2012.

\bibitem{CDH13}
Derek~G. Corneil, Barnaby Dalton, and Michel Habib.
\newblock {LDFS}-based certifying algorithm for the minimum path cover problem
  on cocomparability graphs.
\newblock {\em SIAM J. Comput.}, 42(3):792--807, 2013.

\bibitem{BigArt}
Derek~G. Corneil, J\'er\'emie Dusart, Michel Habib, and E.~K\"ohler.
\newblock On the power of graph searching for cocomparability graphs.
\newblock {\em {SIAM} J. Discrete Math.}, 30(1):569--591, 2016.

\bibitem{COS299}
Derek~G. Corneil, Stephan Olariu, and Lorna Stewart.
\newblock Linear time algorithms for dominating pairs in asteroidal triple-free
  graphs.
\newblock {\em {SIAM} J. Comput.}, 28(4):1284--1297, 1999.

\bibitem{DOS09}
Derek~G. Corneil, Stephan Olariu, and Lorna Stewart.
\newblock The \textsc{{LBFS}} structure and recognition of interval graphs.
\newblock {\em SIAM J. Discrete Math.}, 23(4):1905--1953, 2009.

\bibitem{Todinca1}
Christophe Crespelle and Ioan Todinca.
\newblock An ${O}(n^2)$-time algorithm for the minimal interval completion
  problem.
\newblock {\em Theor. Comput. Sci.}, 494:75--85, 2013.

\bibitem{DP02}
Brian~A. Davey and Hilary~A. Priestley.
\newblock {\em Introduction to Lattices and Order (2. ed.)}.
\newblock Cambridge University Press, 2002.

\bibitem{DS88}
P.M. Dearing, D.R. Shier, and D.D. Warner.
\newblock Maximal chordal subgraphs.
\newblock {\em Discrete Applied Mathematics}, 20(3):181 -- 190, 1988.

\bibitem{JD}
J\'er\'emie Dusart.
\newblock {\em Graph Searches with Applications to Cocomparability Graphs}.
\newblock PhD thesis, Universit\'{e} Paris Diderot, June 2014.

\bibitem{DusH13}
J\'er\'emie Dusart and Michel Habib.
\newblock A new {LBFS}-based algorithm for cocomparability graph recognition.
\newblock {\em to appear in Discrete Applied Mathematics}, 2016.

\bibitem{DM41}
Ben Dushnik and E.~W. Miller.
\newblock Partially ordered sets.
\newblock {\em American Journal of Mathematics}, 63(3):pp. 600--610, 1941.

\bibitem{Fishburn85}
P.C. Fishburn.
\newblock {\em Interval orders and interval graphs}.
\newblock Wiley, 1985.

\bibitem{GH64}
P.C. Gilmore and A.J. Hoffman.
\newblock A characterization of comparability graphs and of interval graphs.
\newblock {\em Canad. J. Math.}, 16:539--548, 1964.

\bibitem{GOL}
Martin~Charles Golumbic.
\newblock {\em Algorithmic Graph Theory and Perfect Graphs (Annals of Discrete
  Mathematics, Vol 57)}.
\newblock North-Holland Publishing Co., Amsterdam, The Netherlands, The
  Netherlands, 2004.

\bibitem{G68}
G.~Gr\"atzer.
\newblock {\em General Lattice Theory}.
\newblock Birkh\"auser, 1968.

\bibitem{HMPV00}
Michel Habib, Ross~M. McConnell, Christophe Paul, and Laurent Viennot.
\newblock Lex-{BFS} and partition refinement, with applications to transitive
  orientation, interval graph recognition and consecutive ones testing.
\newblock {\em Theor. Comput. Sci.}, 234(1-2):59--84, 2000.

\bibitem{HM91}
Michel Habib and Rolf~H. M\"{o}hring.
\newblock Treewidth of cocomparability graphs and a new order-theoretic
  parameter.
\newblock {\em Order}, 11(3):47--60, February 1991.

\bibitem{HMPR91}
Michel Habib, Michel Morvan, Maurice Pouzet, and Jean-Xavier Rampon.
\newblock Extensions intervallaires minimales.
\newblock In {\em Compte Rendu \`a l'Acad\'emie des Sciences Paris,
  pr\'esent\'e en septembre 91, par le Pr. G. Choquet}, volume 313, pages
  893--898, 1991.

\bibitem{HMPR92}
Michel Habib, Michel Morvan, Maurice Pouzet, and Jean-Xavier Rampon.
\newblock Incidence structures, coding and lattice of maximal antichains.
\newblock Technical Report 92-079, LIRMM, 1992.

\bibitem{Todinca2}
Pinar Heggernes, Karol Suchan, Ioan Todinca, and Yngve Villanger.
\newblock Minimal interval completions.
\newblock In Gerth~St{\o}lting Brodal and Stefano Leonardi, editors, {\em ESA},
  volume 3669 of {\em Lecture Notes in Computer Science}, pages 403--414.
  Springer, 2005.

\bibitem{Jakubik91}
J\'an Jakub\'ik.
\newblock Maximal antichains in a partially ordered set.
\newblock {\em Czechoslovak Mathematical Journal}, 41(1):75--84, 1991.

\bibitem{KohlerM14}
Ekkehard K{\"{o}}hler and Lalla Mouatadid.
\newblock Linear time {LexDFS} on cocomparability graphs.
\newblock In {\em Algorithm Theory - {SWAT} 2014 - 14th Scandinavian Symposium
  and Workshops, Copenhagen, Denmark, July 2-4, 2014. Proceedings}, pages
  319--330, 2014.

\bibitem{KrSt93}
Dieter Kratsch and Lorna Stewart.
\newblock Domination on cocomparability graphs.
\newblock {\em SIAM J. Discrete Math.}, 6(3):400--417, 1993.

\bibitem{Lekkerkerker1962}
C.~Lekkerkerker and J.~Boland.
\newblock Representation of a finite graph by a set of intervals on the real
  line.
\newblock {\em Fundamenta Mathematicae}, 51(1):45--64, 1962.

\bibitem{Mark1}
George Markowsky.
\newblock The factorization and representation of lattices.
\newblock {\em Transactions of the American Mathematical Society},
  203:185--200, 1975.

\bibitem{Mark2}
George Markowsky.
\newblock Primes, irreductibles and extremal lattices.
\newblock {\em Order}, 9:265--290, 1992.

\bibitem{MS99}
Ross~M. McConnell and Jeremy Spinrad.
\newblock Modular decomposition and transitive orientation.
\newblock {\em Discrete Mathematics}, 201(1-3):189--241, 1999.

\bibitem{Meister05}
Daniel Meister.
\newblock Recognition and computation of minimal triangulations for {AT}-free
  claw-free and co-comparability graphs.
\newblock {\em Discrete Applied Mathematics}, 146(3):193--218, 2005.

\bibitem{MC}
George~B. Mertzios and Derek~G. Corneil.
\newblock A simple polynomial algorithm for the longest path problem on
  cocomparability graphs.
\newblock {\em SIAM J. Discrete Math.}, 26(3):940--963, 2012.

\bibitem{APPMM}
Rolf~H. M\"ohring.
\newblock Algorithmic aspects of comparability graphs and interval graphs.
\newblock In Ivan Rival, editor, {\em Graphs and Order}, volume 147 of {\em
  NATO ASI Series}, pages 41--101. Springer Netherlands, 1985.

\bibitem{MohRo1996}
Rolf~H. M\"{o}hring.
\newblock Triangulating graphs without asteroidal triples.
\newblock {\em Discrete Appl. Math.}, 64(3):281--287, February 1996.

\bibitem{Parra}
Andreas Parra.
\newblock {\em Structural and algorithmic aspects of chordal graph embeddings}.
\newblock PhD thesis, Technische Universität Berlin, 1996.

\bibitem{Reuter91}
K.~Reuter.
\newblock The jump number and the lattice of maximal antichains.
\newblock {\em Discrete Mathematics}, 88:289--307, 1991.

\bibitem{S02}
B.~S.~W. Schr\"oder.
\newblock {\em Ordered sets, an introduction}.
\newblock Birkh\"auser, 2002.

\bibitem{Trotter99}
William~T. Trotter.
\newblock Combinatorial aspects of interval orders and interval graphs.
\newblock {\em Electronic Notes in Discrete Mathematics}, 2:153, 1999.

\end{thebibliography}
\end{document}